\documentclass{amsproc}

\usepackage{graphicx}
\usepackage{makecell}

\newtheorem{theorem}{Theorem}[section]
\newtheorem{lemma}[theorem]{Lemma}

\theoremstyle{definition}
\newtheorem{definition}[theorem]{Definition}

\theoremstyle{remark}

\numberwithin{equation}{section}

\newcommand{\Asub}{A_{\rm{sub}}}

\newcommand{\Aobj}{A_{\rm{obj}}}
\newcommand{\bea}{\begin{eqnarray}}
\newcommand{\eea}{\end{eqnarray}}

\begin{document}

\title{Mathematics of Floating 3D Printed Objects}

\author[D.M. Anderson]{Daniel M. Anderson}
\address{Department of Mathematical Sciences, George Mason University, Fairfax, Virginia 22030}
\email{danders1@gmu.edu}

\author[B.G. Barreto-Rosa]{Brandon G. Barreto-Rosa}
\address{Department of Mathematical Sciences, George Mason University, Fairfax, Virginia 22030}
\email{bbarreto@gmu.edu}

\author[J.D. Calvano]{Joshua D. Calvano}
\address{Department of Mathematical Sciences, George Mason University, Fairfax, Virginia 22030}
\email{jcalvano@gmu.edu}

\author[L. Nsair]{Lujain Nsair}
\address{Department of Mathematical Sciences, George Mason University, Fairfax, Virginia 22030}
\email{lnsair@gmu.edu}

\author[E. Sander]{Evelyn Sander}
\address{Department of Mathematical Sciences, George Mason University, Fairfax, Virginia 22030}
\email{esander@gmu.edu}

\date{January 1, 1994 and, in revised form, June 22, 1994.}



\begin{abstract}
We explore the stability of floating objects through mathematical
modeling and experimentation. Our models are based on standard ideas of
center of gravity, center of buoyancy, and Archimedes' Principle.  We
investigate a variety of floating shapes with two-dimensional cross sections and identify
analytically and/or computationally a potential energy landscape that
helps identify stable and unstable floating orientations.  We compare
our analyses and computations to experiments on floating objects
designed and created through 3D printing.  In addition to our results,
we provide code for testing the floating configurations for new shapes, 
as well as giving details of the methods for 3D printing the objects.
The paper includes conjectures and open problems for further study.
\end{abstract}

\maketitle

\tableofcontents


\section{Introduction}

Interest in the dynamics of icebergs has been driven by the desire to describe various natural phenomena (including their rolling) as well as
practical considerations associated with shipping and protection of offshore structures in arctic environments.
Allaire's \cite{Allaire1972} study of iceberg stability was motivated by the need to assess mitigation strategies, such as towing of icebergs, to reduce threats posed by icebergs to offshore structures.
Allaire identified readily-identifiable above-water characteristics, such as the ratio of waterline width to above-water height, to estimate stability for a menagerie of iceberg shapes -- Blocky, Drydock, Dome, Pinnacled, Tabular, Growler.
Bailey \cite{Bailey1994}, motivated by similar concerns, examined stability of icebergs in terms of rolling frequency.
The potential for iceberg stability considerations to be dynamic even in calm water due to underwater melting/dissolution was considered by Deriabyn \& Hjorth \cite{DH2009}.
They were also interested to identify what practical, above-water, observations could be made to predict stability changes driven by underwater changes in iceberg morphology.

Ship design and the design of other man-made floating objects has no doubt driven much scientific
and technical work in this field.
We make no attempt to review this literature but interested readers may find resources in the work
of M\'{e}gel \& Kliava \cite{MK2010} and Wilczynski \& Diehl \cite{WD1995}.
Historically speaking, scientific thought on this dates back to Archimedes (c.~287--212/211 B.C.).
See for example Rorres \cite{Rorres2004}, a 
fascinating article on the original work of Archimedes and extensions thereof.

Whereas icebergs and ships are complex three dimensional floating
objects, the floating objects that are the focus of the present work are
those whose configurations can effectively be characterized in
two-dimensions. Specifically, we shall consider `long' objects whose
cross sections are constant for the full length of the object. 
Such shapes have been the
focus of popular online Apps, such as Iceberger~\cite{Iceberger} and the remixed
version~\cite{IcebergerRemix}. 
These
effective two-dimensional floating objects have proven to be
mathematically tractable yet rich in observable phenomena. Prediction of
the stable orientations for a long beam with square-cross section and
uniform density, for example, was considered by Reid \cite{Reid1963}.
Depending on the ratio of the density of the object to the density of
the fluid any rotation of the square can be a stable floating
orientation (i.e.~any orientation from flat side up to corner up). This
configuration has been recently revisited both experimentally and
analytically by Feigel and Fuzailov \cite{FF2021}. These authors
validated experimentally that for a small range of density ratios near
$0.25$ (and also similarly near $0.75$) the full range of stable
orientations can be realized.  Another view of this which we discuss in
more detail in our work is that owing to the four-fold symmetry of the
square, there can be either four or eight stable orientations of the
floating square depending on the density ratio.  We explore how these
results change for what we denote as `off-center' squares.

The case of a long `floating plank' of rectangular cross section was the
focus of work by Delbourgo \cite{Delbourgo1987}. Here, in addition to
the density ratio, another parameter -- the aspect ratio of the
rectangle -- appears. Delbourgo identified in this density ratio vs.
aspect ratio space the existence of six characteristic floating
configurations in terms of (1) long side up or short side up, (2) top
side parallel or not parallel to the waterline, and (2) the number of
submerged vertices. Further studies have explored other cross sectional
shapes and investigated details of the breaking of left-right symmetries
of the floating shapes as the density ratio is varied (e.g.~see
Erd\"{o}s, Schibler, \& Herndon \cite{Erdos_etal1992_part1} for square
and equilateral triangle cross sections and Erd\"{o}s, Schibler, \&
Herndon \cite{Erdos_etal1992_part2} for the three-dimensional shapes of
the cube, octahedron, and decahedron). Work in this area has focused on
homogeneous objects with uniform density.  We shall relax this
assumption in our present work to include some special classes of
non-uniform density in which the center of gravity of the square no
longer resides at its centroid (cf. Definition~\ref{defcog}). 

An excellent review of many important mathematical ideas -- Archimedes'
Principle, Center of Gravity, Center of Buoyancy, and the notion of
Metacenter -- related to floating objects is the work of Gilbert
\cite{Gilbert1991}. One particularly useful concept is that of the
potential energy of floating objects.   Gilbert shows that if one
can compute the potential energy landscape as a function of all possible
orientations of the object one can identify stable floating
configurations by the locations of local minima of the potential energy
function.  This is one of the main objectives of our computations as we
then use this to identify stable orientations.

The present work shares some of the spirit of the paper by Feigel and
Fuzailov \cite{FF2021} to revisit these questions both theoretically and
experimentally.  Our experiments, however, are conducted with 3D
printed shapes.  While some of our theoretical effort has been on objects with square cross
section, our methodology is motivated by the
recognition that 3D printing offers the opportunity to float objects
whose cross sections are effectively limited only by ones own creativity
in defining new shapes.   An example shape we analyze is the `Mason M' shown in Figure~\ref{fig-intro_fig_MasonM}.

\begin{figure}[tb]
\includegraphics[width=0.8\textwidth]{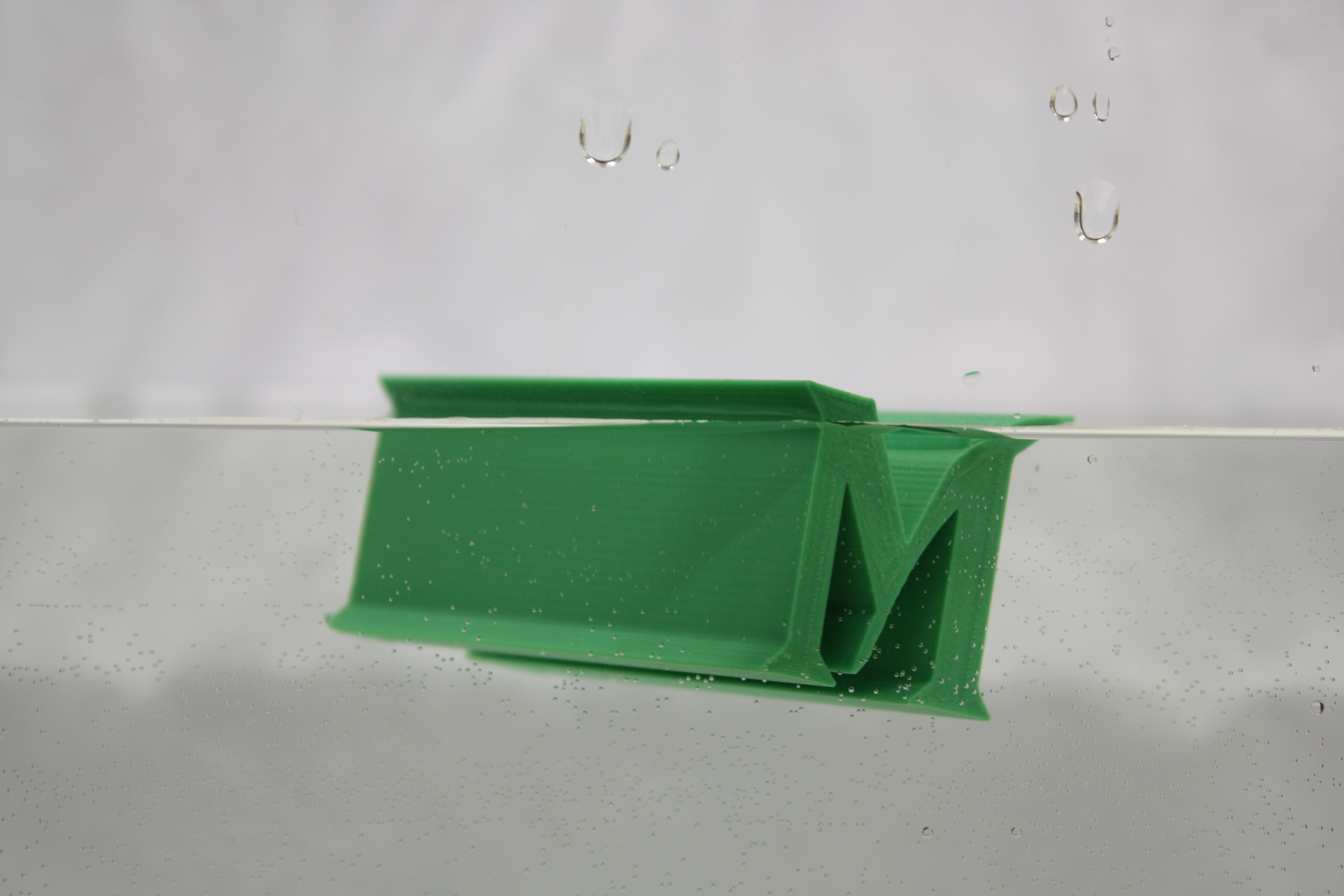}
\caption{The floating Mason M. }
\label{fig-intro_fig_MasonM}
\end{figure}
%

The new contributions we make here are (1) predicting floating configurations of an object 
with square cross section when the center of gravity is not at the object's centroid, 
(2) predicting floating configurations for objects of general two-dimensional polygonal cross section,
and (3) verifying the predictions of (1) and (2) as well as classical ones for objects of square 
cross section with experiments conducted with 3D printed objects.

This paper is organized as follows. In Section~\ref{sec:math}, we introduce definitions and 
terminology, outline prior results for long floating objects of uniform density with 
square cross section, and state our analytical results for the case of long floating 
objects with square cross section with non-uniform density. 
In Section~\ref{sec:3DPrinting} we give a detailed description of how to 3D print floating objects. 
This section contains sufficient details and accompanying code so that interested readers would be 
able to print their own objects and perform their own experiments.
Note that we have additionally provided codes and files in a GitHub repository~\cite{GITHUB}. 
In Section~\ref{sec:Experiments} we describe how we obtain experimental results on 
our 3D printed floating objects.
In Section~\ref{sec:compexp} we describe our code for computing 
stable floating configurations for objects with general polygonal cross sections.
In Section~\ref{sec:Results} we describe the results of our floating experiments and 
relate them to our theoretical results for several cases involving square cross sections
and a selected case of a nontrivial cross section. 
In Section~\ref{sec:conclusion} we give our conclusions and discuss some open problems.


\section{Mathematical  Models for Floating Objects}
\label{sec:math}

\subsection{Definitions}

In this section, we give some basic definitions and concepts that will be needed for 
the remainder of our analysis.

\begin{definition}[Center of Gravity]\label{defcog}
If the mass distribution is given by a continuous 
 density function of $\rho(x,y,z)$ within a domain $\Omega$, 
 then the center of gravity can be obtained by     
 $$(G_x,G_y,G_z)=\frac{1}{M_{\rm{obj} }}\iiint_{\Omega}(x,y,z) \, \rho(x,y,z) \; dV \;,$$
where $M_{\rm{obj}}$ is the object's mass.
In the case of uniform density, i.e. $\rho$ is a constant independent of $x,y$ and $z$, 
the center of gravity is called the centroid. 
\end{definition}
For a long object of length $L$ with uniform cross section, 
 with uniform density in the long direction, i.e. $\rho(x,y,z) = \rho(x,y)$  is independent of $z$, 
the center of gravity is given by
\[
\vec{G} = (G_x,G_y)= \frac{L}{M_{\rm{obj} }}  \iint_\Omega (x,y) \rho(x,y) \; dA 
\]
with $G_z = L/2$ (relative to one end of the object). 
\begin{lemma}
For an object of length $L$ with a polygonal cross-section   with uniform
constant density $\rho$, we can compute the area and center of gravity  as
sums involving only the vertices of the polygon. In particular, let 
\[\{ (x_1,y_1),\dots,(x_N,y_N),(x_1,y_1) \}
\] 
be the vertices of the cross section polygon,
 oriented  counterclockwise. Then the mass of the object is 
\[
M_{\rm{obj}} = \rho L A \;, 
\]
where the area $A$  of the polygon of cross section is given by 
\bea\label{greenarea}
A = \frac{1}{2} \sum_{k = 1}^N (x_k + x_{k+1}) (y_{k+1}- y_k) \;, 
\eea
a result also known as the shoelace formula, and the center of gravity $\vec{G} = (G_x,G_y)$ is given by 
\begin{eqnarray}\label{greencenter}
G_x &=& \frac{1}{6 A}  \sum_{k = 1}^N  (x^2_k + x_k x_{k+1} + x^2_{k+1}) (y_{k+1}- y_k) \\
G_y &=& \frac{1}{6 A} \sum_{k = 1}^N  - (y^2_k + y_k y_{k+1} + y^2_{k+1}) (x_{k+1}- x_k)  \, .  \nonumber 
\end{eqnarray}

\end{lemma}

\begin{proof}
According to Green's theorem
\[
\iint_\Omega \left( \frac{dg}{dx}-\frac{df}{dy} \right) \; dA=\oint_{d\Omega}\, f \; dx+g \; dy,
\]
Choose $(f,g) = (0,x)$ so that $\frac{dg}{dx}-\frac{df}{dy}=1$.  We can parametrize 
the line segment between $(x_k,y_k)$ and $(x_{k+1},y_{k+1})$ by 
 $(x,y) = (1-t)(x_k,y_k)+t(x_{k+1},y_{k+1})$ where $t\in(0,1)$. 
For this line segment, we get  $dx=(x_{k+1}-x_{k}) dt$ and
$dy=(y_{k+1}-y_k) dt$. Now we evaluate the integral from Green's theorem:
\begin{eqnarray*}
\int_0^1 x \; dy &=& \int_0^1(x_k+t(x_{k+1}-x_k))(y_{k+1}-y_k) \; dt \\
&=& x_k(y_{k+1}-y_k)+\frac{1}{2}(x_{k+1}-x_k)(y_{k+1}-y_k) \\
&=& \frac{1}{2}(x_{k+1}+x_k)(y_{k+1} -y_k) \; . 
\end{eqnarray*}
Now we use this to calculate the full area: 
\[
A=\iint_\Omega \;
dA=\oint_{d\Omega} x \; dy=\frac{1}{2}\sum_{k=1}^N(x_{k+1}+x_k)(y_{k+1}-y_k) \;. 
\]
Now we turn our
attention to the calculation of $G_x$. We use Green's theorem 
but this time with $(f,g) = (0,x^2/2)$
in order to satisfy $\frac{dg}{dx}-\frac{df}{dy}=x$. Again for the line segment 
from $(x_k,y_k)$ to $(x_{k+1},y_{k+1})$
we get 
\begin{eqnarray*}
\iint_{\Omega}xdA &=& \oint_{d\Omega}\frac{x^2}{2} \; dy \\
&=&\frac{1}{2}\int_0^1(x_k+t(x_{k+1}-x_k))^2(y_{k+1}-y_k) \; dt \\
&=&\frac{1}{6}(x_k^2+x_k x_{k+1}+x_{k+1}^2)(y_{k+1}-y_k) \; .  
\end{eqnarray*}
Since $\rho$ is constant, we get that  $G_x$ is equal to $\rho L/M_{\rm obj}$ times this
integral, but $M_{\rm obj} = \rho L A$, and this gives the factor of $1/A$ written 
in the formula above. 
In a similar way, using Green's theorem with 
$(f,g) = (-y^2/2,0)$, we get 
\[ \iint_{\Omega}ydA = -\frac{1}{6}(y_k^2+y_k y_{k+1}+y_{k+1}^2)(x_{k+1}-x_k) \; . \]
This gives the equation for $G_y$ stated above.  
\end{proof}

The proof above relies on Green's theorem, but 
there are other ways to derive this formula, 
such as dividing the object into triangles or trapezoids, and combining the 
corresponding triangle or trapezoid area and center of gravity formulas.

\begin{definition}[Buoyancy]
 Buoyancy is a force exerted on an object that is wholly or partially submerged
 in a fluid.  The magnitude of this force is equal to the weight of the displaced fluid.
 Buoyancy relates to the density of the fluid, the volume of the displaced
 fluid, and the gravitational field; it is independent of the mass and density
 of the immersed object. The buoyancy force acts vertically upward at the
 centroid of the displaced volume.  The {\em center of buoyancy} is given by
 \begin{equation}
 (B_x,B_y,B_z)  = \frac{1}{V_{\rm{sub}}} \iiint_{\Omega_{\rm{sub}}} (x,y,z) \;  dV,
\end{equation}
where $V_{\rm{sub}}$ is the submerged volume of the object, and $\Omega_{\rm{sub}}$ is the 
submerged domain. Like in the case of center of gravity,  in the uniform
cross section case $B_z = L/2$. We use the notation $\vec{B} = (B_x,B_y)$.  
Note that $\vec{B}$ is the centroid of the submerged domain. 
\end{definition}
We now state one of the major results necessary for understanding floating objects.
\begin{theorem}[Archimedes' Principle]
The upward buoyant force exerted on an object wholly or partially
submerged is equal to the weight of the displaced fluid.  
In the absence of other forces, such as surface tension, this can be expressed in 
the force balance as
\bea \label{eq:buoy}
M_{\rm{obj}} g & = & \rho_{\rm{f}} V_{\rm{sub}} g,
\eea
where $g$ is acceleration due to gravity, $\rho_{\rm{f}}$ is the density of the fluid, 
and $V_{\rm{sub}}$ is the submerged volume of the object.
\end{theorem}

Observe that~\eqref{eq:buoy} represents a balanced (net zero) equation of
competing forces  with the left terms representing gravitational force and right
term representing the opposing force of buoyancy. Note that if the object has uniform
density, then the mass of the
object can be written as $M_{\rm{obj}} = \rho_{\rm{obj}} V_{\rm{obj}}$.  
In this case, it follows that
\bea
\frac{V_{\rm{sub}}}{V_{\rm{obj}}} & = & \frac{\rho_{\rm{obj}}}{\rho_{\rm{f}}}.
\eea
For our purposes, Archimedes' Principle determines the appropriate waterline
intersections defining a submerged volume whose value relative to the total volume 
matches the appropriate density ratio.  However, it
is important to note that satisfying Archimedes' Principle is not a sufficient
condition for determining a stable equilibrium. An equilibrium orientation of a
floating body occurs when the center of gravity and the center of buoyancy are
vertically aligned. If $\vec{G}$ lies directly below $\vec{B}$ the equilibrium is stable,
whereas if $\vec{G}$ lies
above $\vec{B}$ the equilibrium may or may not be stable. We present an alternative
approach using energy principles similar to that of Erd{\"o}s~\cite{Erdos_etal1992_part1} 
and Gilbert~\cite{Gilbert1991}. After
identifying a waterline that is consistent with Archimedes' Principle, we define a unit
vector normal to the waterline, by keeping the object fixed and rotating the frame of
reference (waterline) by angle $\theta$ to generate all orientations
satisfying Archimedes' Principle. The stable positions of a floating body
occur at the minima of the potential energy. The  
potential energy function for a floating body  is given by 
\bea \label{eq:PEfunction}
U(\theta) =  \hat{n}(\theta)\cdot (\vec{G}-\vec{B}(\theta)) ,
\eea
where $\hat{n}$ is the unit normal vector to the waterline pointing out of
the water, and $\theta$ is the rotation angle of the waterline. Note that $\vec{B}(\theta)$ and
$\hat{n}(\theta)$ depend on $\theta$, but $\vec{G}$ is independent of $\theta$.  
Below, we derive formulas for the stable floating configurations by 
finding minima for $U(\theta)$. 




\subsection{Square Cross Section Revisited}\label{subsec:squarerevisit}



The stability of a long floating object with square cross
section, and the corresponding nontrivial floating configurations, have been investigated theoretically in a number of studies.
Reid \cite{Reid1963} provided the first theoretical identification
of stable floating equilibrium configurations based on arguments using forces and
moments.  Feigel \& Fuzailov \cite{FF2021} provided a recent alternative derivation
of these equilibrium conditions, a brief review of related studies, and also 
detailed experiments validating the theory.  Their experiments had a particular focus on floating configurations
in the transition from `flat side up' orientations to `corner up' orientations.
We revisit the square cross section configuration here with the goal of
writing down the entire potential energy landscape, whose minima reveal the stable equilibrium
configurations.    The square has four-fold symmetry, and we
exploit this in the identification of a center of buoyancy
formula.  A new contribution we make in the present work corresponds to situations in which the center of
gravity is not at the center of the square.  We specifically explore the breaking
of this four-fold symmetry for floating objects
with square cross sections and use the corresponding potential energy landscapes to 
understand the observations.   In sections that follow, we demonstrate that the identification
of potential energy landscapes can be obtained for shapes of more general cross sections
in order to understand their stable floating configurations.

Rather than fix a waterline and consider different orientations of the
square we fix a reference frame on the square with corners at $(1,-1)$,
$(1,1)$, $(-1,1)$, $(-1,-1)$ and consider different orientations of the
waterline. Three configurations are relevant as shown in
Figure~\ref{fig-square_diagram} -- the first has the waterline
intersecting opposite sides of the square and the second and third have the
waterline intersecting adjacent sides of the square.  We work out these
three cases below and then give the generalization for all orientations.

For our $[-1,1]^2$ square, the cross sectional area is $\Aobj = 4$.  If
we denote by $\Asub$ the submerged area, Archimedes' Principle requires
that
\bea
\label{archimedes}
\frac{\Asub}{\Aobj} = R,
\eea
where $R \in (0,1)$ is the density ratio $\rho_{\rm{obj}}/\rho_{\rm{f}}$
of the floating object to the fluid.  We shall assume that the object's
density is uniform throughout but if it were not the appropriate
interpretation
of $\rho_{\rm{obj}}$ for the application of Archimedes' Principle would
be the effective density -- i.e. the object's mass divided by the volume
of the object.  Note that in the present context we work in terms of
cross sectional area; corresponding volumes would be obtained by
multiplying the cross sectional area by the length of the object in the
third dimension.

Below we outline the computation of the center of buoyancy, $\vec{B}(\theta)$, as a function of orientation $\theta$ for the two cases in which (1) the waterline intersects opposite sides of the square and (2) the waterline intersects adjacent sides of the square.

\begin{figure}[h!]
\begin{center}
\includegraphics[height=0.35\textwidth]{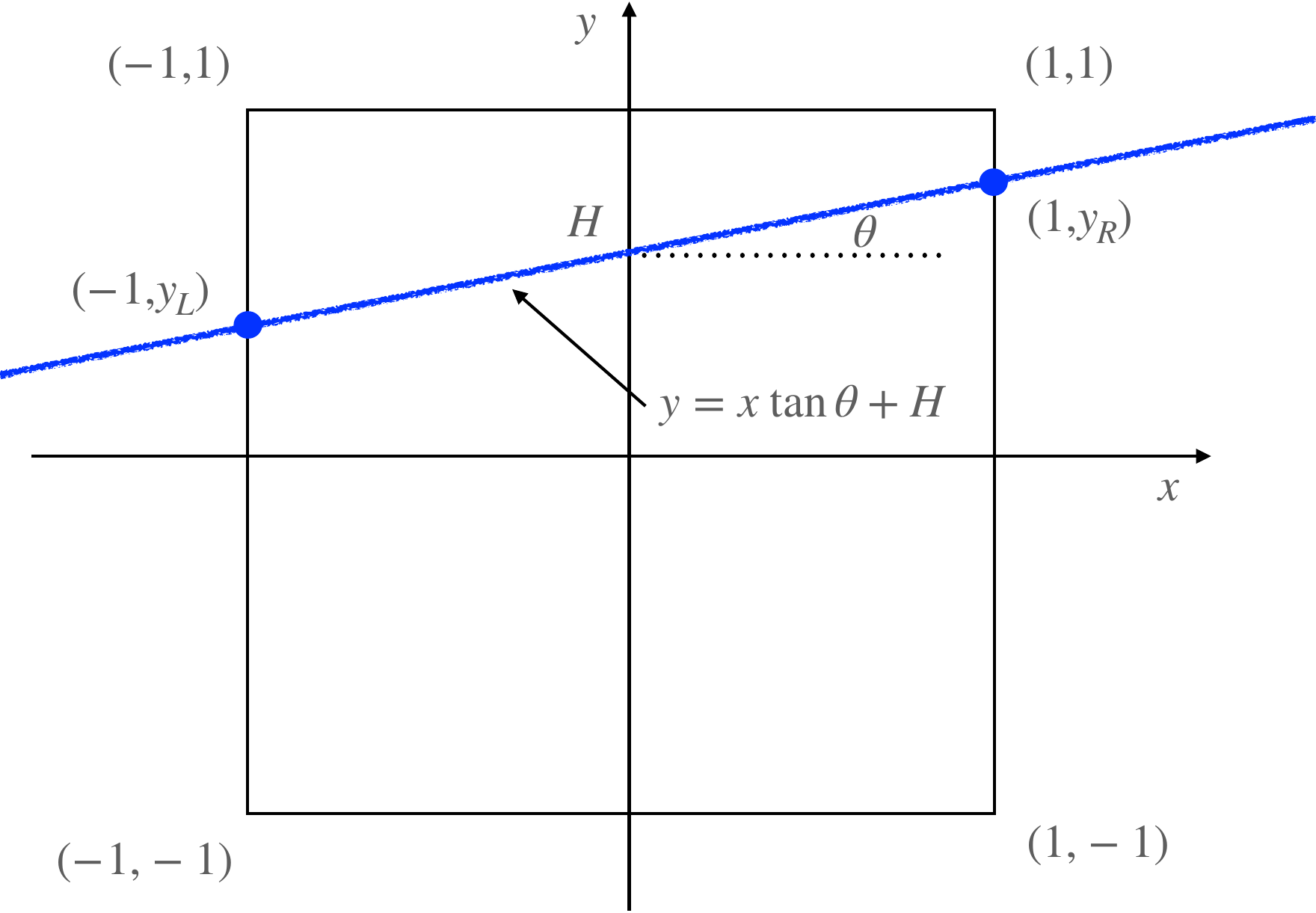}
\includegraphics[height=0.35\textwidth]{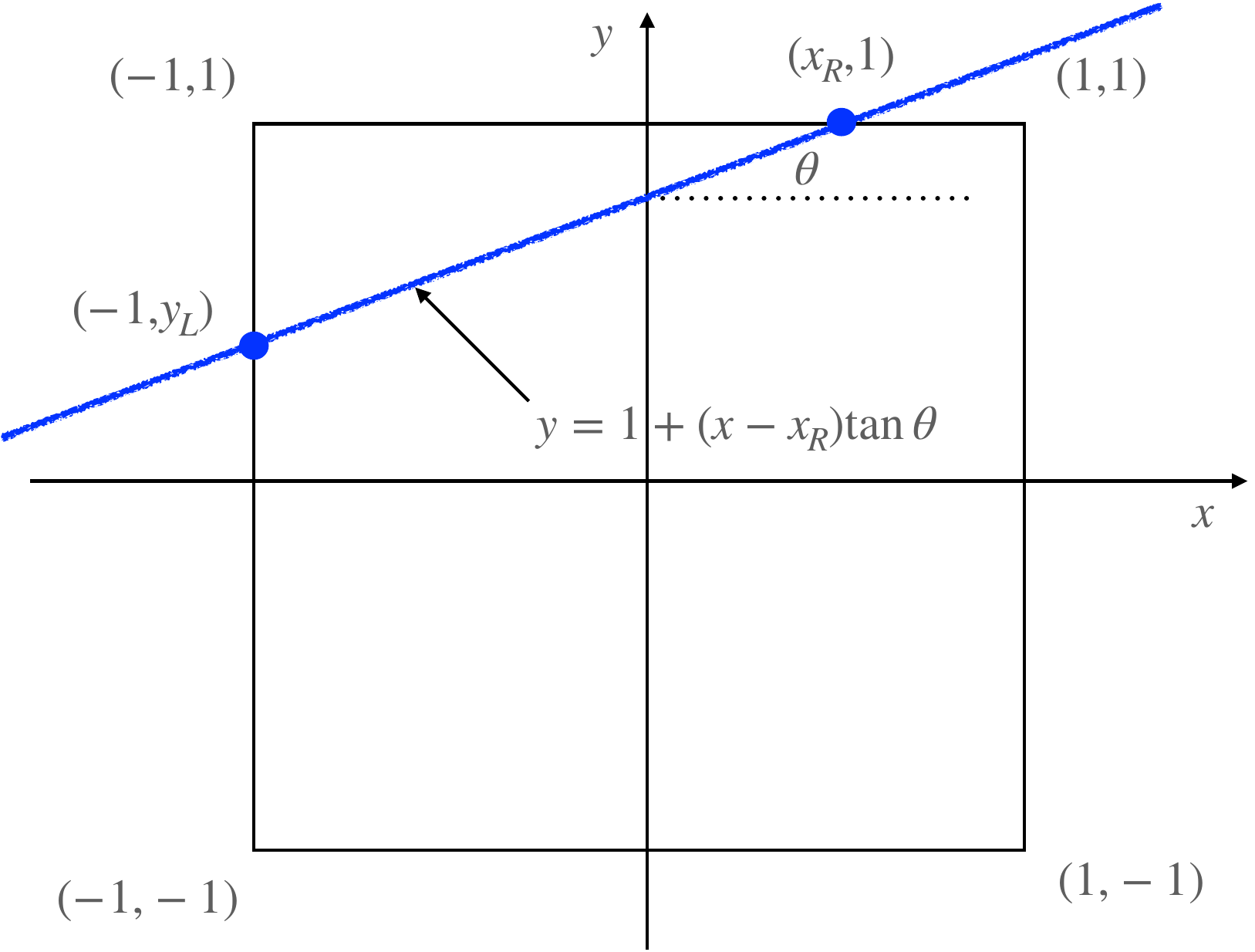}\\ \vspace{0.05in}
\includegraphics[height=0.37\textwidth]{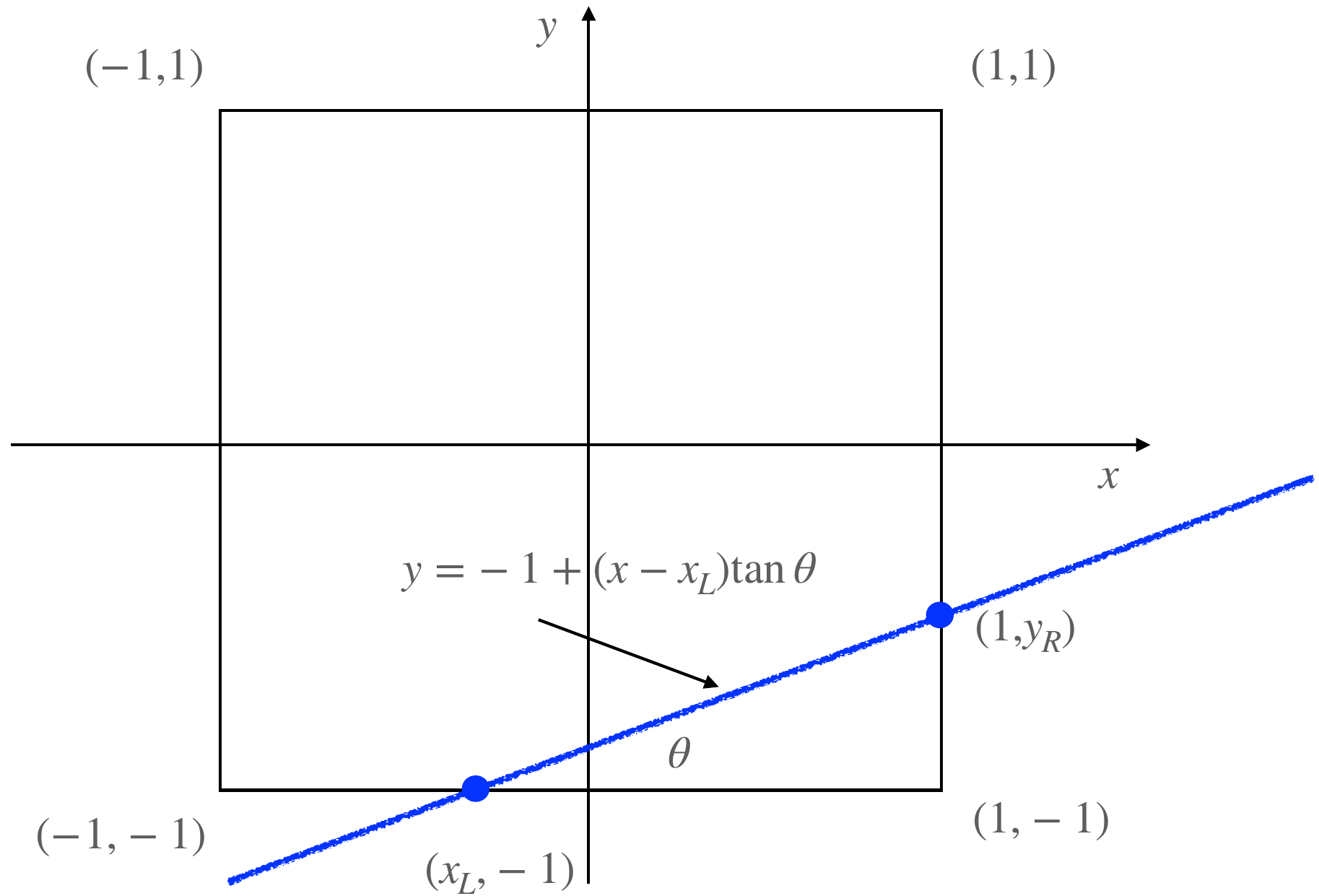}
\end{center}
\caption{The sketch on the upper left shows the configuration in which the waterline
(blue) intersects opposite sides of the square. The sketch on the upper right
shows the configuration in which the waterline intersects adjacent sides
of the square for $R> 1/2$.  The sketch on the bottom
shows the configuration in which the waterline intersects adjacent sides
of the square for $R< 1/2$.}
\label{fig-square_diagram}
\end{figure}


\subsubsection{Waterline Intersects Opposite Sides of Square}

Here we define the waterline by the equation
\bea
\label{eq:waterline_case1}
y & = & x \tan \theta + H,
\eea
where $\theta$ is the slope of the waterline and $H$ is the y-intercept
(see Figure~\ref{fig-square_diagram}).   With the water assumed to occupy
the region below the waterline, the submerged area can be written in
terms of $H$ as
$\Asub = 2 (1 + H)$.
Therefore, Archimedes' Principle requires $R = (1+H)/2$,  or equivalently $H = 2R - 1$. 
Note that for $R \in(0,1)$ it follows that $H \in (-1,1)$.

We define waterline intersection points $(-1,y_L)$ and $(1,y_R)$ and note that 
\bea
y_L = - \tan \theta + H, \quad
y_R = \tan \theta +H.
\eea
By definition, the configuration under consideration requires that 
$y_L \in [-1,1]$ and $y_R \in [-1,1]$. Furthermore, the largest and smallest 
$\theta$ occur  for $y_R = \pm 1, y_L = \mp 1$, meaning that 
$-\pi/4 \le \theta \le \pi/4$ and $-1 \le \tan \theta \le 1.$ 
Combining these facts with the definitions of $y_L$ and $y_R$, we get 
\bea
\tan \theta + H \le 1 \hspace{0.25in} \mbox{and} \hspace{0.25in} -\tan \theta + H \le 1 \hspace{0.25in} \mbox{if $H \ge 0$}, \\
-\tan \theta + H  \ge -1 \hspace{0.25in} \mbox{and} \hspace{0.25in} \tan \theta + H \ge -1 \hspace{0.25in} \mbox{if $H \le 0$},
\eea
These can be rewritten as 
\bea
\label{eq:tan_theta_A}
 -1 + | H | \le \tan \theta \le 1 - | H |.
\eea
This range of $\tan \theta$ corresponds to  a range of $\theta$ values 
$[\theta_1^{\min},\theta_1^{\max}]$ contained in  $[-\pi/4,\pi/4]$. By symmetry, 
there is a corresponding configuration when rotated by $\pm \pi/2$ and $\pm \pi$.

The submerged area in this configuration is defined by the four points $(1,-1)$, $(1,y_R)$, $(-1,y_L)$, and $(-1,-1)$.  
We use \eqref{greencenter} to find the center of buoyancy as the centroid of the 
submerged boundary region.  
In particular, let $(x_k,y_k)$ be given by $\{ (1,-1), (1,y_R), (-1,y_L),(-1,-1),(1,-1) \}$. 
Then $\vec{B}(\theta) = (B_x(\theta),B_y(\theta))$ where
\begin{eqnarray*} 
 B_x(\theta) & = & \frac{1}{6 \Asub} \sum_{k=1}^{4} (x_k^2 + x_k x_{k+1} + x_{k+1}^2)(y_{k+1}-y_k)\\
B_y (\theta) & = & \frac{1}{6 \Asub} \sum_{k=1}^{4} -(y_k^2 + y_k y_{k+1} + y_{k+1}^2)(x_{k+1}-x_k) \;. 
\end{eqnarray*}
Computing these sums, combined with the values of $y_L$ and $y_R$ and the fact that $\Asub = 2(1+H)$, 
we find that the center of buoyancy takes the form
\bea
\vec{B}(\theta) = \vec{B}_1(\theta) & \equiv & \frac{1}{2 (1+H)} \Big( \frac{2}{3} \tan \theta , -1 + H^2 + \frac{1}{3} \tan^2 \theta \Big),
\eea
where $\theta$ can take on any value defined by the inequalities~(\ref{eq:tan_theta_A}).  For use below we define this specific form for the center of buoyancy as $\vec{B}_1(\theta)$.

\subsubsection{Waterline Intersects Adjacent Sides of Square}

A similar approach can be applied to the second configuration shown in the upper right sketch of Figure~\ref{fig-square_diagram}.
Here we give expressions for the results when $R \ge 1/2$ and when $R < 1/2$.

{\bf Case 1: $R \ge 1/2$.}
Here we assume that the waterline intersects the left and top sides of the square at points $(-1,y_L)$ and $(x_R,1)$ so that three corners of the square are submerged.  
We consider the case $R< 1/2$ in the next section although note that this case can be carefully extracted from the
present case (according to Gilbert \cite{Gilbert1991}, Feigel \& Fuzailov \cite{FF2021}, among others).

Here we identify the waterline by
\bea
y & = & (x - x_R) \tan \theta  + 1,
\eea
where $y_L = - (1 + x_R) \tan \theta  + 1$. 

This waterline cuts a triangular region of area $\frac{1}{2} (1+ x_R) (1- y_L)$ from the original square.  This means that
$\Asub = 4 - \frac{1}{2} (1+x_R) (1- y_L)$ and 
\bea
R = \frac{\Asub}{\Aobj} & = & \frac{ 4 - \frac{1}{2} (1+x_R) (1- y_L)}{4} =  \frac{ 4 - \frac{1}{2} (1+x_R)^2 \tan\theta}{4}.
\eea
Rearranging this gives $x_R$ in terms of $R$ and the waterline slope $\tan \theta$ 
\bea
(1+ x_R)^2 & = & \frac{8 - 8 R}{\tan \theta}.
\eea

Conditions on $\theta$ come from the requirement that $0 \le (1+x_R) \le 2$ and $-1 \le y_L \le 1$.    The first of these reveals that
\bea
0 \le  \frac{2 - 2 R}{\tan \theta} \le 1.
\eea
For the case under consideration $0 < \theta < \frac{\pi}{2}$.  It follows that $\tan \theta \ge 2 - 2R$.
Equality corresponds to the waterline passing through the point $(1,1)$ and $(-1,y_L)$.
The other extreme corresponds to the waterline passing through the point $(-1,-1)$ and $(x_R,1)$.  This has
$\tan \theta = 2/(1+x_R)$.
Here
the triangular area is $\frac{1}{2} 2 (1+x_R)$ which means $R = (4 - (1+x_R))/4$ or $(1+x_R) = 4 - 4R$.  Since $\theta$
cannot exceed this angle we have $\tan \theta < 2/(4-4R)$.  Thus, for this configuration the value of $\tan \theta$ is 
constrained by
\bea \label{theta2minmax}
2 - 2R \le \tan \theta \le \frac{1}{2-2R}.
\eea

As in the previous case, 
the center of buoyancy is obtained by calculating the centroid
of the submerged area  using \eqref{greencenter}. In particular, using 
the counterclockwise oriented  
vertices of the submerged polygon: 
\[
\{(x_R,1),(-1,y_L),(-1,-1),(1,-1),(1,1)\} , 
\]
we can calculate the following integral 
\bea
\vec{B}(\theta) & = & \frac{1}{\Asub } \iint_{\Omega_{\rm sub}} (x,y) \; dA
\eea
as a sum. 
Evaluating this integral gives
\bea
\vec{B}(\theta) = \vec{B}_2^{+}(\theta) & \equiv & \frac{1}{\Asub} \Big( 1 - \frac{1}{2} (y_L + 1)  - \frac{1}{6} (x_R^3 + 1) \tan\theta ,   \nonumber \\
                 & & \mbox{}   -1 + \frac{1}{2} (1-x_R)  + \frac{(1-y_L^3)}{6\tan\theta} \Big).
\eea
For use below, we define this specific form for the center of buoyancy as $\vec{B}_2^{+}(\theta)$. 
Recall that $\Asub = 4R$ and
\bea
y_L = - \tan \theta (1 + x_R) + 1,\quad
(1+ x_R)^2 = \frac{8 - 8 R}{\tan \theta},
\eea
where $\tan \theta$ satisfies \eqref{theta2minmax}.
This condition gives a range of 
$\theta$ values given by $[\theta_2^{+ \min},\theta_2^{+ \max}]$ contained in $[0,\pi/2]$. 
Again by symmetry, we get a corresponding set of angles by adding $\pm \pi/2, \pm \pi$. 

{\bf Case 2:  $R < 1/2$.} Here assume that the waterline intersects the square at points $(x_L,-1)$ and $(1,y_R)$ so that only the lower right corner of the square is submerged.  

Here we identify the waterline by
\bea
y & = & \tan \theta (x - x_L) - 1,
\eea
where $y_R =  \tan \theta (1 - x_L) - 1$.

This waterline cuts a triangular region of area $A_{sub} = \frac{1}{2} (1+ y_R) (1- x_L)$ from the original square.  This means that
\bea
R = \frac{\Asub}{\Aobj} & = &\frac{1}{8} (1+ y_R) (1- x_L) = \frac{(1-x_L)^2 \tan\theta}{8}.
\eea
Rearranging this gives $x_L$ in terms of $R$ and the waterline slope $\tan \theta$ 
\bea
(1- x_L)^2 & = & \frac{8 R}{\tan \theta}.
\eea

Conditions on $\theta$ come from the requirement that $0 \le (1- x_L) \le 2$.    This translates to 
\bea
\tan \theta \ge 2 R.
\eea
Also the condition  $-1 \le y_R \le 1$ leads to 
\bea
0 \le \tan \theta \le \frac{2}{1 - x_L}.
\eea
Using $8R = (1-x_L) (1+ y_R)$ leads to $0 < \tan \theta < 1/(2R)$.

So, together these require
\bea \label{theta2minus}
2R \le \tan \theta \le \frac{1}{2R}.
\eea

As before, the center of buoyancy satisfies
\bea
\vec{B}(\theta) & = &  \frac{1}{\Asub }\int_{\Omega_{\rm sub}} (x, y) \; dA , 
  %
  %
    %
    %
    %
 \eea
 which can be calculated as a sum involving the vertices of the submerged polygonal 
 cross section via ~\eqref{greencenter}. 
It follows that
\bea
\vec{B}(\theta) &=& \vec{B}_2^{-}(\theta)  \\
& \equiv & \frac{1}{\Asub} \Big( \frac{y_R+1}{2}  - \frac{(1-x_L^3)}{6} \tan \theta ,
  - \frac{(1-x_L)}{2}  + \frac{1+y_R^3}{6\tan \theta} \Big) \; .\nonumber
\eea
For use below, we define this specific form for the center of buoyancy as $\vec{B}_2^{-}(\theta)$. 
Recall that $\Asub = 4R$ and
\bea
y_R =  \tan \theta (1 - x_L) - 1,\quad
(1- x_L)^2 & = & \frac{8 R}{\tan \theta},
\eea
where $\tan \theta$ satisfies \eqref{theta2minus}.

\subsubsection{Potential Energy Expressions: Square Cross Section}

As defined in ~\eqref{eq:PEfunction} the potential energy function is given by
\[
U(\theta)  =  \hat{n}(\theta) \cdot (\vec{G} - \vec{B}(\theta)) ,
\]
where the unit normal to the waterline can be expressed as a function of
$\theta$ as $\hat{n}(\theta) = ( - \sin \theta, \cos \theta )$. For the square
defined above with uniform density the center of gravity
$\vec{G}=(0,0)$.    However, we are interested in a generalization of
the square where the center of gravity, by some means, is not
necessarily located at the center but rather has coordinates $\vec{G} =
(G_x, G_y)$.  Note that as long as Archimedes' Principle is applied with
the appropriate mass of the object, the calculations presented above for
the center of buoyancy are independent of the location of the center of
gravity.  So, in what follows we treat $\vec{G}$ as nonzero in general. 

For $R \ge 1/2$ define
\bea
U_{B_1}(\theta) & = &  \hat{n}(\theta) \cdot \vec{B}_1(\theta) \, ,  \mbox{ for } -1 + | H | \le \tan \theta \le 1 - | H | \, ,
\eea
which corresponds to a range of $\theta \in [\theta_1^{\min}, \theta_1^{\max}]$  
defined by~\eqref{eq:tan_theta_A}. 
Also define
\bea
U_{B_2^+}(\theta) & = &  \hat{n}(\theta) \cdot \vec{B}_2^+(\theta) \, , \mbox{ for } 2 - 2R \le \tan \theta \le \frac{1}{2-2R} \, ,
\eea
which corresponds to a range of $\theta \in [\theta_2^{+ \min}, \theta_2^{+ \max}]$ 
defined by~\eqref{theta2minmax}. 

We can write the potential energy function $U(\theta)$ as follows
\bea
\label{eq:PE_formula_SQUARE}
U(\theta) & = & \left\{
\begin{array}{ll}
\hat{n}(\theta) \cdot \vec{G} - U_{B_1}(\theta) & \theta \in [\theta_1^{\min}, \theta_1^{\max}] \\
\hat{n}(\theta) \cdot \vec{G} - U_{B_1}(\theta \pm \frac{\pi}{2}) &  \theta \pm \frac{\pi}{2} \in [\theta_1^{\min}, \theta_1^{\max}] \\
\hat{n}(\theta) \cdot \vec{G} - U_{B_1}(\theta \pm \pi) &  \theta \pm \pi \in [\theta_1^{\min}, \theta_1^{\max}] \\
& \\
\hat{n}(\theta) \cdot \vec{G} - U_{B_2^+}(\theta) & \theta \in [\theta_2^{+ \min}, \theta_2^{+ \max}] \\
\hat{n}(\theta) \cdot \vec{G} - U_{B_2^+}(\theta \pm \frac{\pi}{2}) &  \theta \pm \frac{\pi}{2} \in [\theta_2^{+ \min}, \theta_2^{+ \max}] \\
\hat{n}(\theta) \cdot \vec{G} - U_{B_2^+}(\theta \pm \pi) &  \theta \pm \pi \in [\theta_2^{+ \min}, \theta_2^{+ \max}] \\
\end{array}
\right.
\eea
A similar formula applies when $R < 1/2$ (replace $B_2^+$ with $B_2^-$ and the 
corresponding range for $\tan \theta$ given in~\eqref{theta2minus}).

\subsection{Squares With Off-Center Weights}

We also explore the case of a square cross section with an off-center weight parallel to the long axis of the object.  Specifically we consider
3D printed objects with a hole in the square that can be filled with a material of different density.  In our experiments we had the option to leave
the hole as void space or to insert a nail cut to fit the object.   In either case, before floating the object tape was placed over the holes to prevent 
water from filling the space.

For such a configuration we can predict the modified center of gravity $\vec{G} \neq 0$.  In particular, consider the same square with corners
at $(1,-1)$, $(1,1)$, $(-1,1)$, $(-1,-1)$ with a hole with circular cross section at point $(a,b)$ with $a \in (0,1)$, $b \in (0,1)$, and radius $r_H$.  When such an object
is printed there is a border around the hole whose thickness we denote by $t$ and whose density is $\rho_{\rm PLA}$ (i.e.~the density of the solid print material).  
With the hole filled with a nail whose density
is $\rho_{\rm nail}$ we can compute the center of gravity of the object as a whole (printed object plus nail) as
\bea
\label{eq:COG_hole1}
M_{\rm{obj}} \vec{G} & = & L \left\{ \int_{\Omega_0} \rho(\vec{x}) \vec{x} \; dA + \int_{\Omega^{{\rm hole}+{\rm border}}} \rho(\vec{x}) \vec{x} \; dA \right\},
\eea
where $M_{\rm{obj}}$ is the mass of the whole object (including the nail if one is inserted), $L$ is the length of the object, $\Omega_0$ denotes the cross-sectional domain of the square
excluding the hole and border and $\Omega^{{\rm hole}+{\rm border}}$ denotes the circular cross section that includes the (printed) border of the hole and 
the hole, and $\rho(\vec{x})$ denotes the material density at position $\vec{x}$ in the plane.  
The square printed without a hole will have some void space in its interior and this can be controlled by changing the infill of the print.  In our
squares printed with a hole this infill region gets replaced by the hole plus the border material of the hole.  Therefore, it is convenient to rewrite
the formula~(\ref{eq:COG_hole1}) for center of gravity $\vec{G}$ as 
\bea
\label{eq:COG_hole2}
\frac{M_{\rm{obj}}}{L} \vec{G} & = & \int_{\Omega_0} \rho(\vec{x}) \; \vec{x} \;  dA + \int_{\Omega^{{\rm hole}+{\rm border}}} \rho(\vec{x}) \; \vec{x} \; dA \nonumber \\
& & \mbox{} + \int_{\Omega^{{\rm hole}+{\rm border}}} \rho_{\rm infill} \;  \vec{x} \; dA - \int_{\Omega^{{\rm hole}+{\rm border}}} \rho_{\rm infill} \; \vec{x} \; dA, \nonumber \\
& = & \int_{\Omega} \rho(\vec{x}) \; \vec{x} \; dA +  \int_{\Omega^{{\rm hole}+{\rm border}}} (\rho(\vec{x}) - \rho_{\rm infill}) \;  \vec{x} \; dA,
\eea
where $\Omega$ denotes the cross section of the square undisturbed by a hole.  Under our assumption of a square with uniform density the 
first integral in this expression equates to the zero vector.  That is, for a square without the off-center hole the center of gravity is located at $(0,0)$.
This requires that the infill is sufficiently symmetric about the center of the square so that it negligibly moves the center of
gravity away from $(0,0)$.\footnote{This appears to be a good approximation for the {\em grid} infill pattern but not, for example, the {\em cat} infill pattern or for {\em grid} at very low infills.} 
It follows that for the off-center square the center of gravity can then be estimated as
\bea
\label{eq:COG_hole3}
\frac{M_{\rm{obj}}}{L} \vec{G} & = & \int_{\Omega^{{\rm hole}+{\rm border}}} (\rho(\vec{x}) - \rho_{\rm infill}) \; \vec{x} \; dA, \nonumber \\
 & = &(\rho_{\rm nail} - \rho_{\rm infill})  \int_{0}^{2\pi} \int_0^{r_H} \vec{x} r\; dr \; d\theta \nonumber \\
  & & \mbox{} + (\rho_{\rm PLA} - \rho_{\rm infill}) \int_{0}^{2\pi} \int_{r_H}^{r_H+t} \vec{x} r \; dr \; d\theta,
\eea
where each of the density terms in these expressions are assumed to be independent of position.  These integrals can be evaluated writing
$\vec{x} = (a + r \cos \theta, b + r \sin \theta)$.  It follows that
 \bea
\label{eq:COG_hole4}
\frac{M_{\rm{obj}}}{L} \vec{G}  & = & \pi (\rho_{\rm nail} - \rho_{\rm infill}) r_H^2 (a,b) 
+ \pi (\rho_{\rm PLA} - \rho_{\rm infill}) [ (r_H+t)^2 - r_H^2] (a,b),\nonumber \\
& = & \left\{ \pi(\rho_{\rm nail} - \rho_{\rm infill}) r_H^2
+ \pi (\rho_{\rm PLA} - \rho_{\rm infill}) [ 2 r_H t + t^2 ] \right\} (a,b).
\eea
So, for the off-center square with hole at $(a,b)$ the center of gravity is shifted towards $(a,b)$ from the origin by terms proportional to density differences and
cross-sectional areas.  Note that all of the quantities in this expression can be determined by straightforward measurements and are listed in Table \ref{Table-OffCenter}.

Practically speaking, our prints are not completely uniform in the direction orthogonal to the square face since the top and bottom square faces are solid PLA.
An improved estimate for the center of gravity that accounts for the two end faces of the square of thickness $t$ with density $\rho_{\rm PLA}$ is
 \bea
\label{eq:COG_hole4b}
M_{\rm{obj}} \vec{G} & = & \Big\{ \left[ (\rho_{\rm nail} - \rho_{\rm infill}) (\pi r_H^2)
+ (\rho_{\rm PLA} - \rho_{\rm infill}) \pi (2 r_H t + t^2) \right] (L-2t)  \nonumber \\
& & \mbox{} + 2 t (\rho_{\rm nail} - \rho_{\rm PLA} ) (\pi r_H^2)  \Big\} (a,b), \nonumber \\
& = & M_{\rm nail} (a,b) + \Big\{ \left[  - \rho_{\rm infill} (\pi r_H^2)
+ (\rho_{\rm PLA} - \rho_{\rm infill}) \pi (2 r_H t + t^2) \right] (L-2t)  \nonumber \\
& & \mbox{} - 2 t \rho_{\rm PLA}  (\pi r_H^2)  \Big\} (a,b).
\eea
That is, the new center of gravity, $\vec{G}$, is shifted towards the
hole location $(a,b)$ by an amount related to the mass of the nail (we
use $M_{\rm nail}=0$ for an open hole) and terms related to the 
thickness of the hole and the material it replaces (either infill or
 boundary).

 \begin{table}[t]
\caption{Various parameter values for 3D prints with square cross section and a hole.
The values of $M_{\rm nail}$ and $\rho_{\rm nail}$ were obtained by noting that each nail was 60 mm in length and 2 mm in radius and that 25 nails
weighed $137.51$ g.}\label{Table-OffCenter}
%
%
\renewcommand\arraystretch{1.5}
\noindent\[
\begin{array}{|c|c|c|}
\hline
\mbox{Parameter} & \mbox{Description} & \mbox{Value}   \\ \hline
I 
& \makecell[c]{\mbox{Infill Fraction} \\ (\mbox{Infill \% /100})}   
& 0.05 ... 0.95 \\ \hline
s 
& \mbox{Length of Side of Square}
& 30 \; \mbox{mm}\\ \hline
L 
& \mbox{Length of Object}
& 60 \; \mbox{mm}\\ \hline
\rho_{\rm PLA} 
& \mbox{Density of PLA} 
& 1.15 \mbox{ g cm}^{-3}\\ \hline
M_{\rm nail} 
& \mbox{Mass of Nail} 
& 5.5004 \mbox{ g}\\ \hline
\rho_{\rm nail} 
&\makecell[c]{\mbox{Density of Nail}  \\ M_{\rm nail}/{V_{\rm nail}} }
& 7.295 \mbox{ g cm}^{-3} \\ \hline
\rho_{\rm infill} 
& \makecell[c]{\mbox{Effective Density of Infill} \\ \rho_{\rm PLA} \times I} 
& \mbox{varies with Infill} \\ \hline
r_H 
& \mbox{Radius of Hole} 
& 2.5 \; \mbox{mm} \\ \hline
t 
& \mbox{Thickness of Solid Border} 
& 0.8 \; \mbox{mm} \\ \hline
\end{array}
\]
\end{table}

The predicted effective density for the block of square cross section is
\bea
\rho^{\rm eff}_{\rm square} & = & \frac{M_{\rm square} }{V_{\rm ext}},
\eea
where the total mass of the object is  
\bea
M_{\rm square} & = & (V_{\rm ext} - V_{\rm int})\rho_{\rm PLA} + V_{\rm int} I \rho_{\rm PLA},
\eea
and $V_{\rm ext} = s^2 L$ and $V_{\rm int} = (s - 2t)^2  (L - 2t)$.
When the block is printed with a hole of radius $r_H$ parallel to the long axis the predicted effective density is
\bea \label{eq:effective_density_hole}
\rho^{\rm eff}_{\rm square + hole} & = & \frac{1}{V_{\rm ext}} \Big\{ 
M_{\rm square} + M_{\rm nail} - \pi r_H^2 [ 2 t + (L - 2t) I ] \rho_{\rm PLA} \nonumber \\
& & \mbox{}                 +  \pi [ (r_H +t)^2 - r_H^2] (L - 2t) (\rho_{\rm PLA} - I \rho_{\rm PLA})
\Big\}
\eea
The effective density of the object without the nail filling the hole is calculated by the same formula with $M_{\rm nail}$ set to zero.

Various comparisons between this theory and our experimental observations and measurements are given below.  First, however,
we need to create our floating objects.


\section{Methods for 3D Printing}
\label{sec:3DPrinting}
In order to experiment with floating objects, 
we have opted to experiment by designing and 3D printing them. This has 
the advantage that we can easily create any object we can describe mathematically. 
In addition, we can vary the density of our print by changing the infill density, a parameter 
which is set at the time of printing. In this section, we describe the full workflow 
needed -- and we have kept our presentation accessible to those with no 3D printing experience, 
in the hopes of making it possible for anyone with an interest to create their own experiments. 
The process consists of three steps: first, we need to design the objects in 
design software. We then need to  give the  print specific parameters, 
within a {\it slicer} software, where the choice of slicer software depends on the printer. Finally, we
print the objects on a 3D printer. We detail the workflow of these steps here. 

\begin{figure}[tb]
\includegraphics[width=0.5\textwidth]{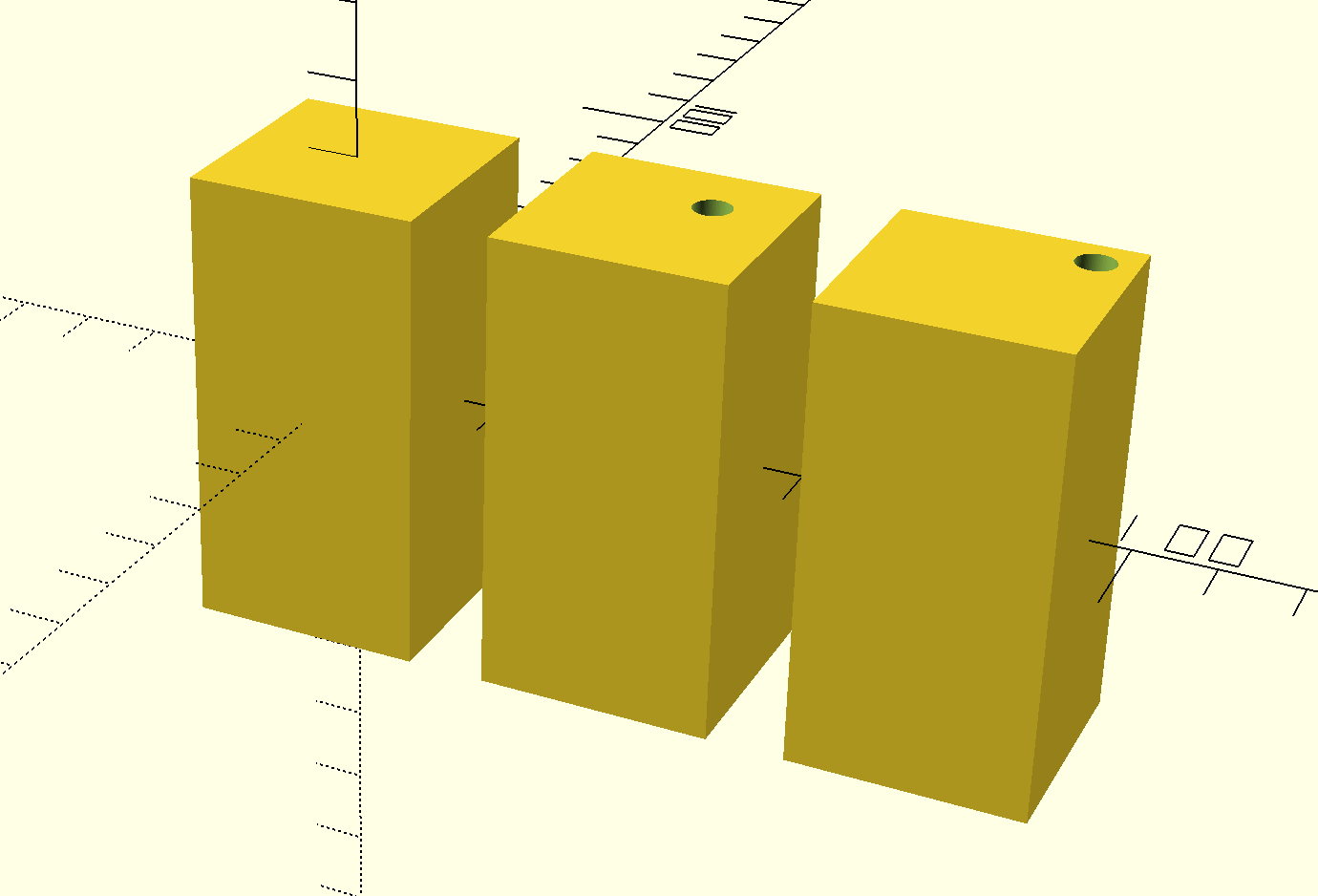}
\caption{This shows three cubes in OpenSCAD. The leftmost cube is at the ``true" center. The 
other two cubes have a hole on the diagonal line of the cross section. All cubes are 
shown with the longest direction vertical. This is the orientation in which we printed all cubes.}
\label{firstfig}
\end{figure}
\subsection{Design} 
We have opted to design the objects in 
 OpenSCAD~\cite{OpenSCAD}, which is a free command
line computer aided design (CAD) system.  Figure~\ref{firstfig} shows 
three of our experimental floating objects generated in OpenSCAD.
Since we are interested in creating objects with a fixed cross section, 
we are able to do so in just a few lines of code. 
For convenience, we include syntax here so a reader could create their own. Copies of both 
sample code and stl files are available from~\cite{GITHUB}.
To create a box with height 60 mm 
and a 30 mm by 30 mm  square cross section, we use the command
\begin{verbatim}
   cube([30,30,60]);
\end{verbatim}
For our off-center weight experiments, we have placed a hole of radius 2.5 mm lengthwise 
in the interior of the box. This is done by taking the set difference between 
the box  above and a cylinder. In order to create a cylinder of height 60 mm and 
radius 2.5 mm, we use the command 
\begin{verbatim}
   cylinder (h = 60, r=2.5, center = true, $fn=100);
\end{verbatim}
The command {\tt center = true} centers the object so it is easier to position it with respect to the cube. 
The command {\$fn=100} indicates that the circle should be estimated by a 100-sided polygon. 
In order to have the cylindrical hole to be displaced from the center of the cube, we use the 
{\tt translate} command. 
Putting this all together  to create a vertical hole that is displaced by 10 mm diagonally
from the center of the cube, we use the command sequence:
\begin{verbatim}
   difference() {
      cube([30,30,60], center =true);
      translate([10,10,0]) 
      	cylinder (h = 80, r=2.5, center = true, $fn=100);
   }
\end{verbatim}

While most of our floating objects were simple cubes with and without holes, 
in Section~\ref{sec:polygoncode} we also considered more 
complicated objects, where the cross section was given as a polygon in the form 
$\{(x_1,y_1),(x_2,y_2),\dots,(x_n,y_n)\}$. 
This includes for example the Mason M in Figs.~\ref{fig-intro_fig_MasonM} \&~\ref{masonfloat}.
 In order to create objects with general polygonal cross sections in OpenSCAD, 
 we first create a polygon with the {\tt polygon} command, and then we 
can turn this into a solid with fixed cross section using the 
{\tt linear\_extrude} command. For example, 
the following creates a cylinder of height 60mm with a cross section which is a  
polygon determined by the ordered points $(0,0),(50,0),(60,40),(50,20)$.
\begin{verbatim}
   points=[[0,0],[50,0],[60,40],[50,20]]; 
   linear_extrude(height=60) polygon(points);
\end{verbatim} 

\subsection{Slicing and Printing} We now describe the printing parameters set 
within the slicing software.
In order to speed up printing and keep objects light, a 3D printed
object is usually partially hollow inside, printed with a lattice
pattern filling a fixed {\em infill fraction} of the object interior.
We denote this quantity by $I$.  The fact that we can vary the infill fraction  is quite
useful for us, since our goal is to see how the stable floating
orientations change with density.  To be consistent, we did not vary
 the {\em infill pattern}; we printed all of our prints with the {\em grid} infill pattern. 
  For further consistency, all of our floating objects were printed on a Makerbot Replicator 
 5th Generation printer and sliced in Makerbot Print proprietary software. 

The mass of a print is most significantly affected by the infill fraction, but this 
is not the only factor. In addition, an outer border layer of the print is printed at 100\% infill
for a fixed thickness on the sides, and a fixed thickness on the top and bottom. 
For the most part, the default values for all thickness is 
0.8 mm, the thickness of two  {\em shells}, 
 where a shell is $0.4$ mm, i.e. the diameter 
 of a standard extruder nozzle. 

In order to calculate the predicted mass of a print, we need to know $\rho_{\rm PLA}$,
the density of PLA. Most sources state that this value is around 1.24 -- 1.25 g cm$^{-3}$, though~\cite{DAetal} used the value 1.17 g cm$^{-3}$. This number can vary quite significantly
depending on the brand of PLA, and therefore we decided to measure the value for our lab conditions. 
In particular, we printed a number of squares that 
were 30 mm by 30 mm by 1.6 mm. Since the 
 the top and bottom of a print are $0.8$ mm solid filament, 
the print is guaranteed to be 100\% infill. 
We printed fifteen test tiles, roughly half with Makerbot brand filament
 and half with an off brand filament. Most of our test tiles were printed in part on 
 the same Makerbot  printer and slicer as the floating objects, but we also printed
 some on a Monoprice Mini. We 
  found that the printer used was not an important factor for the mass of the test tiles, but the 
brand was an important factor.  For Makerbot brand filament and off-brand filament 
respectively,   our measured densities were $1.082 \pm 0.017$ and 
$1.152\pm 0.010$ in g cm$^{-3}$. We have chosen the value $\rho_{PLA} = 1.15$ g cm$^{-3}$ for our 
calculations.

\begin{figure}[tb]
\includegraphics[width=0.7\textwidth]{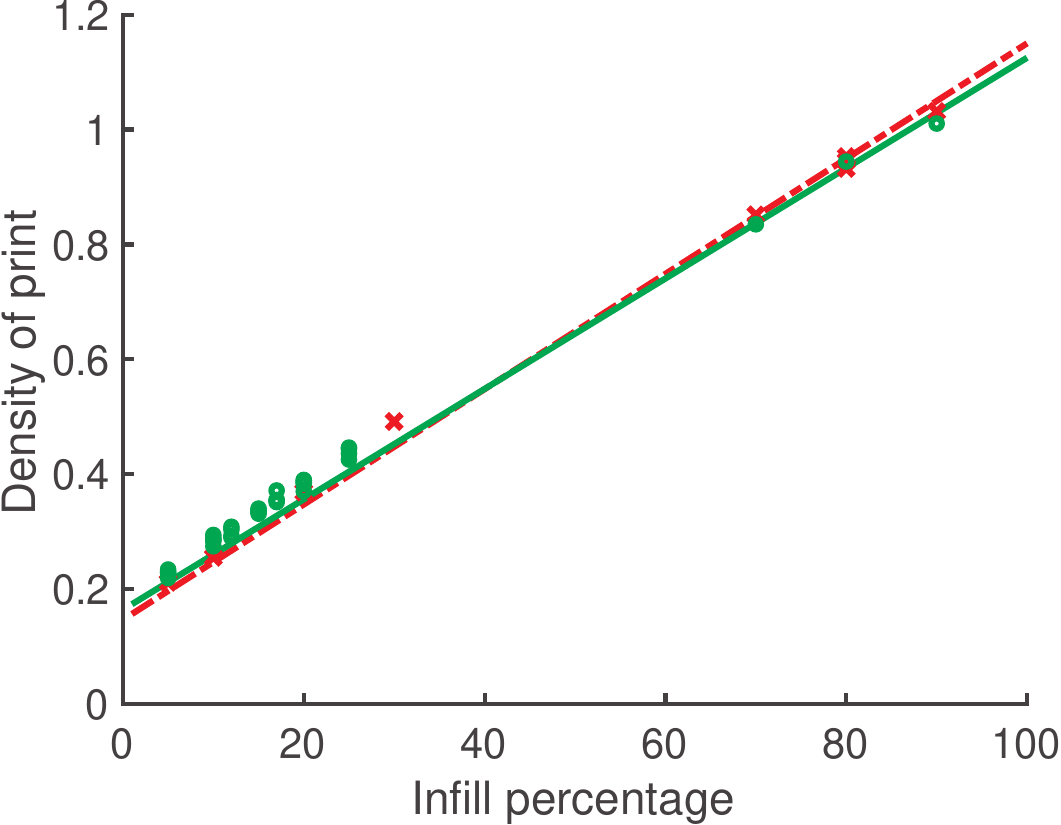}
\caption{Density of prints  as a function of infill. Prints 
have 30mm by 30mm cross section and 
60mm length. The predicted values
without hole (dashed red) and with hole (solid green), graphed 
along with the measured values without hole (red x) and with hole (green o). }
\label{densityest}
\end{figure}

Taking into account the  infill percentage, shells, and measured value
for $\rho_{PLA}$, we expect the mass of our object to be given by
$\rho_{\rm PLA}$ times $V_{\rm ext}$, the volume of the outer layer plus
$\rho_{\rm PLA}$ times $I$, the infill fraction, times $V_{\rm int}$,
the volume of the interior of the print. If there is a nail embedded in
the hole, then we additionally have to include this value in our
calculation. We have described this calculation above in
\eqref{eq:effective_density_hole}. Figure~\ref{densityest} shows the predicted values
for prints without a nail hole (dashed red) with a nail hole (solid green). 
We also show the measured values for the prints we have made without a nail hole
(red x) and with a nail hole (green o). 
Note that while quite useful to be able to 
predict mass of a print from the infill density, all of our calculations of $R$ that 
we used for particular prints in our experiments are based on 
the mass of the print measured directly using a digital scale.

It is critical for the print to be uniform in the longest direction, as
our calculations  consider this direction  as completely invariant.
Therefore we printed with the long direction going from top to bottom so
that the infill is identical at each cross section.

Since our prints were printed on a Makerbot printer, we have used the
Makerbot proprietary slicing software Makerbot Print. However, there is
nothing about the process that could not be modified to other slicer
software.

The objects are ready to float!


\section{Experiments and Data Acquisition}
\label{sec:Experiments}

Once the objects were printed our objectives turned to floating these objects and obtaining measurements associated with the 
stable floating orientations that we could compare with the theory.    
After floating the objects in a test tank to observe qualitative behavior (e.g.~confirming the objects floated, identifying the number of stable orientations, etc.)
we turned to gathering quantitative information in the form of stable floating orientations.

We used a tank with clear flat sides and filled with tap water (see Figure~\ref{experimentalsetup}).
In order to measure angles associated with a stable floating orientation the object was carefully placed by itself in the tank and allowed to 
come to rest.    In order to ease the task of keeping the object in a position with its square cross section facing the side of the tank where the
camera was positioned we placed vertical guides a small distance away from the object.  One of these guides was a ruler that we could later
use in the image analysis for calibration.  A typical view of a floating object is shown in Figure~\ref{fig-typical_experiment}.

\begin{figure}[tb]
\includegraphics[width=0.8\textwidth]{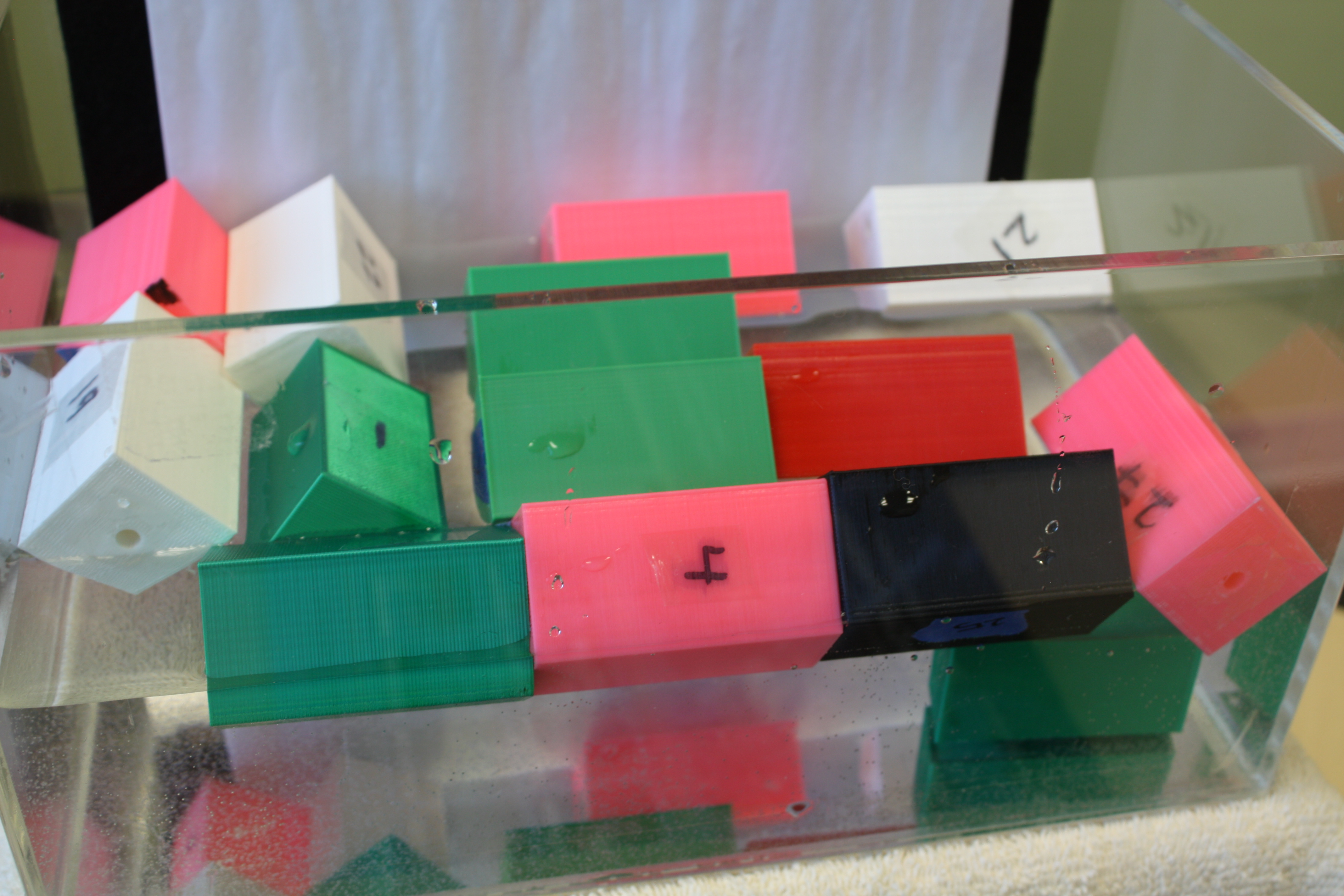}
\caption{Our experimental setup involves floating 3D printed objects mostly with 
square cross sections.  An occasional print had density ratio $R >1$ and did not float. }
\label{experimentalsetup}
\end{figure}

%
\begin{figure}[tb]
\includegraphics[width=0.8\textwidth]{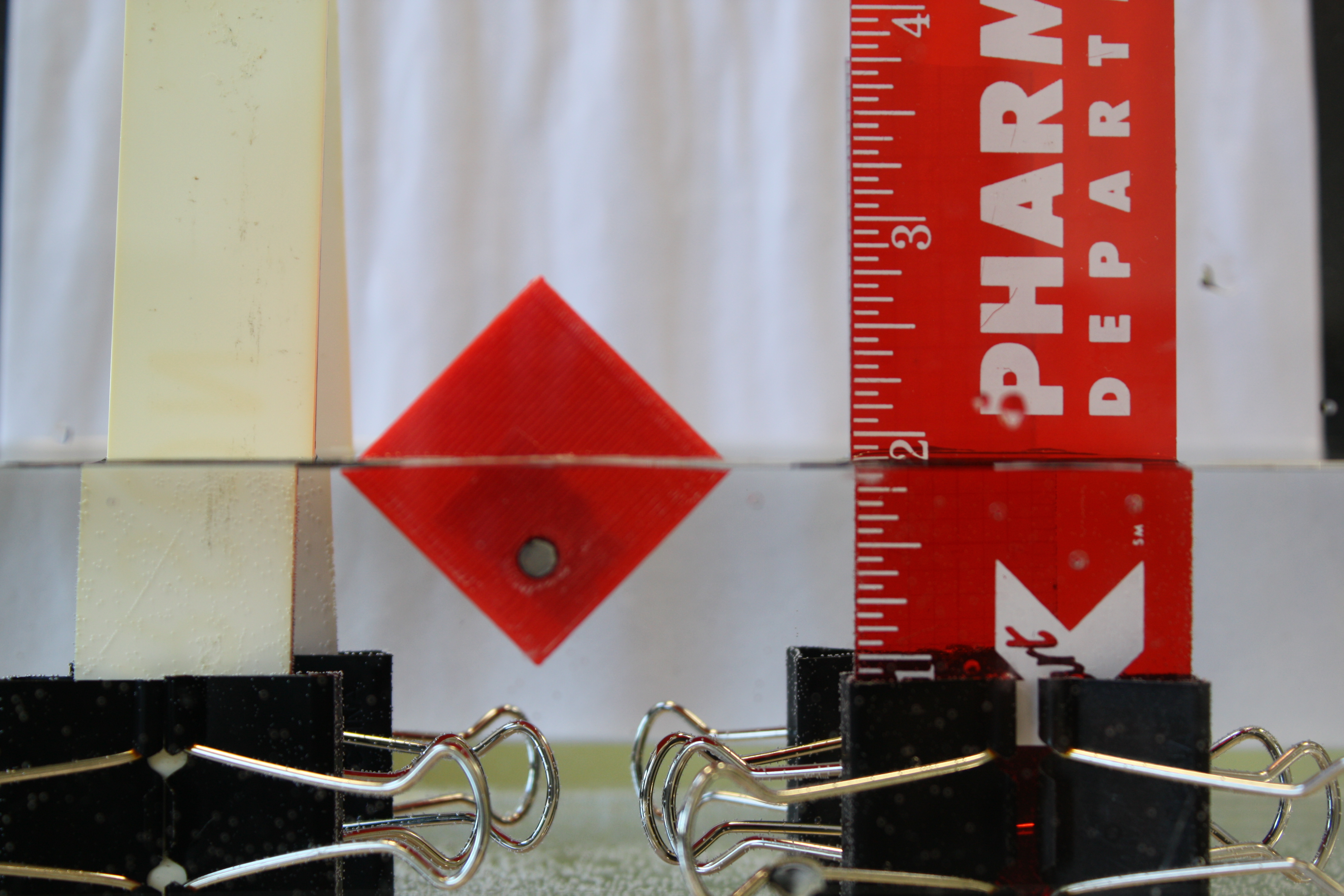}
\caption{A typical view of the experimental set up to measure angles associated with stable floating orientations.  Vertical guides were
used as an aid to keep the object floating in the part of the tank on which the camera was focused.  Care was taken to assure that
contact with underwater objects was avoided and any incidental contact of the object with the guides did not alter its floating orientation.
We also took care to prevent contact of object with the sides of the tank.}
\label{fig-typical_experiment}
\end{figure}
%

Digital images were obtained with a Canon EOS Rebel XSi used in manual focus mode.  A remote shutter release was used to avoid bumping the camera, 
which was mounted on a tripod sitting on the same table as the tank.  Care was also taken to line up the camera at the water level and to have the object floating
in a direct line to the camera.   A certain amount of refraction could be observed, especially in directions off the main viewing axis.

We used {\tt Matlab's} {\tt grabit.m} software to extract information from each image such as the one in Figure~\ref{fig-typical_experiment}.
We first calibrated the view by selecting points on the ruler in view on the right, taking care to use points above the waterline to avoid 
distortion of the ruler by the water.
For a floating square we then identified 6 points in the image.  Two of these points were chosen on the waterline relatively far from the
object on both sides so that the orientation of the waterline could be obtained.  Four other points were chosen at the four corners of the square (clockwise starting from the top of the square).  
With this information we obtained four vectors by taking differences of adjacent vertices and used $\cos \theta = a^T b / \| a \|_2 \|b \|_2$ where $b$ is a waterline vector and 
$a$ is one of square-side vectors.  In the example of Figure~\ref{fig-typical_experiment} one pair of sides (northwest and southeast) were used to identify an estimate for $\theta$ while
the other pair of sides (northeast and southwest) were used to estimate the complementary angle with respect to $\pi/2$.  
In the case of the Mason M, we used a similar approach but chose four points along the bottom two `legs' of the M to make angle measurements.

For prints with holes, we covered the holes (void or filled with a nail) with waterproof tape to keep water out.   In cases with no hole or when the hole was in the center of the object,
we used a very small amount of nail polish to mark one corner of the square.  This allowed us to assess asymmetry of the object that was present either unplanned or by design.

For each object we first measured its mass (with hole either left as a void or filled with a nail) using a digital scale and from that obtained an effective density
according to $\rho_{\rm{obj}} = M_{\rm{obj}} / V_{\rm{obj}}$.  
This corresponds to the material density for a uniform object or the effective uniform density for a non-uniform object.
The volume of the object $V_{\rm{obj}}$ was computed either from measurements of the dimension and known formulas (e.g. width times length times height for a rectangular box)
or using {\tt polyarea} in {\tt Matlab} (for more general cross sections) multiplied by the 3D print scaling.

\section{Computational Approaches}
\label{sec:compexp}

Various {\tt Matlab} codes were developed to compute results and analyze our floating shapes.    

\subsection{Square Cross Sections}

Various {\tt Matlab} codes developed for the square cross sections have been posted in 
a GitHub repository~\cite{GITHUB}. 
These include
\begin{itemize}
\item {\tt SQUARE\_PE\_GxGy.m}: This code is based on the potential energy formulas outlined in the section on the square.  
It generates for given values for the density ratio $R$ and the center of gravity $(G_x,G_y)$ the computed potential 
energy landscape and includes options to plot the potential energy landscape and the square floating in a stable orientation. This code
was used to generate the theoretical predictions in Figures~\ref{fig-square_PE_and_SHAPE_10}, \ref{fig-square_PE_and_SHAPE_31n},
and~\ref{fig-square_PE_and_SHAPE_56}, for example.
\item {\tt SQUARE\_ANGLES\_GxGy\_R\_Looper.m}: This code is based on the formulas given in the section on the square and plots for a given 
center of gravity $(G_x,G_y)$ and specified range $R \in [R_{min},R_{max}]$ the stable orientation angles.  This code, for example, 
was used to generate Figures~\ref{fig-THETA_VS_R_zeroG} and~\ref{fig-THETA_vs_R_32nail}.
\end{itemize}

\subsection{General Polygonal Cross Sections}\label{sec:polygoncode}

In the case of a long floating object of uniform density 
with a general polygonal cross section, 
though it is no longer possible to give as 
detailed an analysis as in the case of a square, we are still able to apply Archimedes' Principle 
and calculate the center of gravity, center of buoyancy, and the potential energy of a floating 
configuration, as we describe below. 

\subsubsection{Computation of Stable Floating Configurations}

    When doing the calculations for a general polygon, we wrote a program that
    takes two vectors with the $x$ and $y$ values of our polygon and a density
    ratio and goes through the following algorithm.  Firstly, it takes the shape
    of the given polygon and calculates the center of gravity of the object,
    assuming a uniform density throughout. Secondly, we identify the correct
    placement of the waterline determined by Archimedes' Principle that establishes
    the correct submerged area to total area ratio. This is done
    by use of a bisection method where at an orientation we take the lowest
    point of our object and create a waterline through it and make that our lower
    bound. We then take the highest point of our object and make that our upper
    bound. We then calculate the area ratio of each of our bounds and find
    the midpoint of our upper and lower bounds and calculate the area ratio
    for the waterline going through the midpoint. If the waterline through the
    midpoint is above the correct placement, it becomes our new upper bound and
    if it is below the correct placement, then it becomes the new lower bound.
    We apply the bisection method until the correct waterline is found
    for our original orientation. The correct placement of our waterline in the
    case of uniform density is where the waterline splits the object into two
    areas where the submerged area relative to the total area 
    is equal to the desired density ratio
    (e.g. for an iceberg with uniform density, the line will be such that the 
    new submerged area created by the line relative to the area of the original polygon
    will match the density ratio $0.8912$). Thirdly, we compute the center of
    buoyancy of the object by applying the same method used for the center of
    gravity except using the submerged area determined by the polygon and the
    waterline. Finally, we calculate the potential energy function for our polygon at the
    orientation and repeat the process for all angles.
    
    The input to this code is a planar polygonal region, oriented counterclockwise and a density 
    ratio of the object relative to the water. 
    The output is a plot of the potential energy landscape with respect to the 
    angle. By default, the computation is performed for uniform density. However, 
    it is possible to compute this information for objects with non-uniform density 
    if one inputs the center of gravity.

\section{Results}
\label{sec:Results}

\subsection{Floating Squares: Symmetric Case}

Figure~\ref{fig-THETA_VS_R_zeroG} shows stable equilibrium angles as a function
of density ratio $R$ for objects with square cross sections and center of
gravity at the center of the square, $\vec{G}=(0,0)$.  Various experimental
results are shown for 3D printed shapes with different effective densities. 
These effective densities have been modified as described earlier by adjusting the infill as well as
printing objects with a hole at the center which we can leave as void space or
fill with a denser object, such as a nail.  

As a visualization, we show several potential energy landscapes and selected
shape orientations predicted from the theory and observed experimentally.
Figure~\ref{fig-square_PE_and_SHAPE_10} shows the potential energy landscape for
$\vec{G}=(0,0)$ for a case with $R=0.23296$ which corresponds to a region in
parameter space where eight stable orientations exist.  The eight orientations
in this case come in pairs, as indicated in the lower portions of  Figure~\ref{fig-square_PE_and_SHAPE_10}.  
The eight experimentally-observed orientations are shown by the points at $R=0.23296$ in 
Figure~\ref{fig-THETA_VS_R_zeroG}.\footnote{A keen eye will note both blue dots and red
dots in this sequence in Figure~\ref{fig-THETA_VS_R_zeroG}.  These two different sets of angle estimates correspond to the two different
angle measurements described in the earlier section on Experiments and Data Aquisition.}
Figure~\ref{fig-square_PE_and_SHAPE_31n} shows the potential energy landscape
for $\vec{G}=(0,0)$ for a case with $R=0.4856$ which corresponds to a region in
parameter space where four stable orientations exist.  These orientations
correspond to the object floating with the corner straight up.
Experimental measurements for these angles correspond the angle measurements shown at $R=0.4856$ in Figure~\ref{fig-THETA_VS_R_zeroG}.
Figure~\ref{fig-square_PE_and_SHAPE_56} shows the potential energy landscape for
$\vec{G}=(0,0)$ for a case with $R=0.9322$ which corresponds to a region in
parameter space where four stable orientations exist.  These orientations
correspond to the object floating with the flat side of the square straight up.
Experimental measurements for these angles correspond the angle measurements shown at $R=0.9322$ in Figure~\ref{fig-THETA_VS_R_zeroG}.

A more thorough experimental exploration of the parameter space, particularly in the region around $R=0.25$, where eight stable orientations can be identified 
has recently been done by Feigel \& Fuzailov \cite{FF2021}.  Those authors used a larger floating object (114 mm $\times$ 114 mm $\times$ 353 mm) constructed using two tin tea 
boxes that could be fitted/weighted with additional bars and magnets to adjust the object's effective density.

\begin{figure}[h!]
\begin{center}
\includegraphics[width=4.0in, angle = 0]{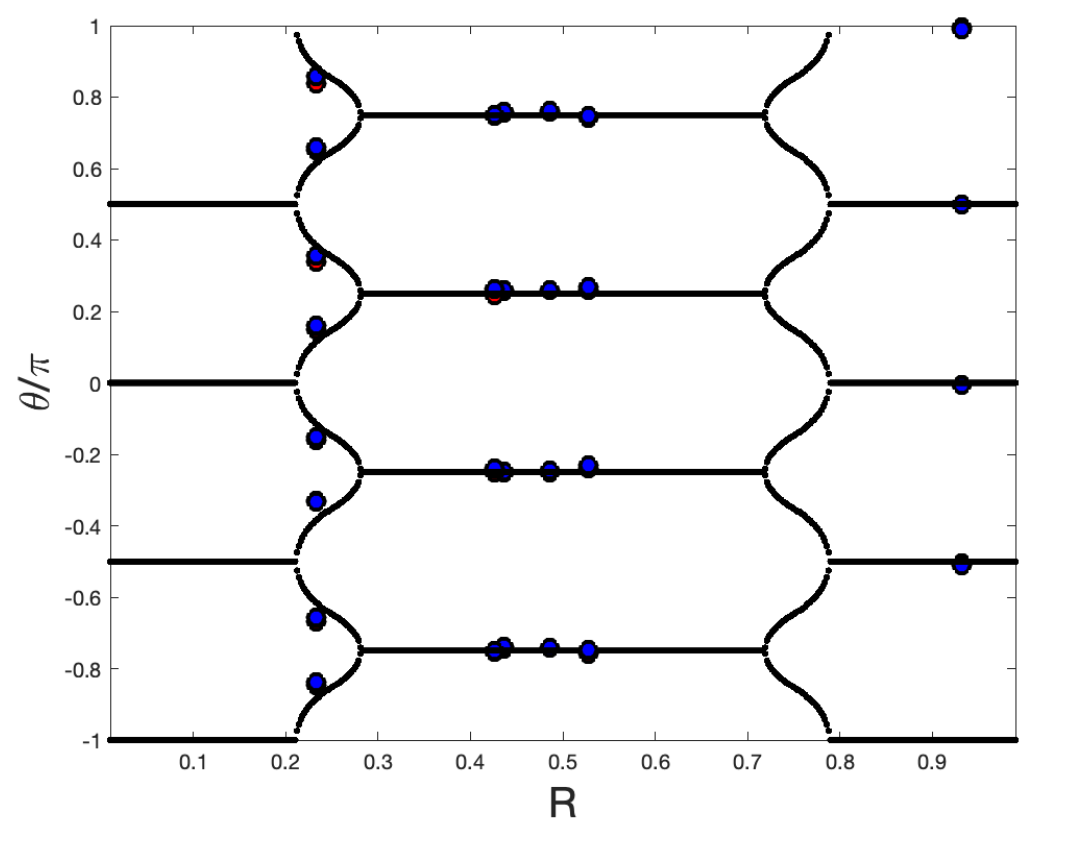}
\end{center}
\caption{This plot shows stable floating orientations versus density ratio $R$ for $\vec{G}=(0,0)$.  The marks show various measured equilibrium orientations for several
of our 3D printed objects.
}
\label{fig-THETA_VS_R_zeroG}
\end{figure}

\begin{figure}[h!]
\begin{center}
\includegraphics[width=2.75in, angle = 0]{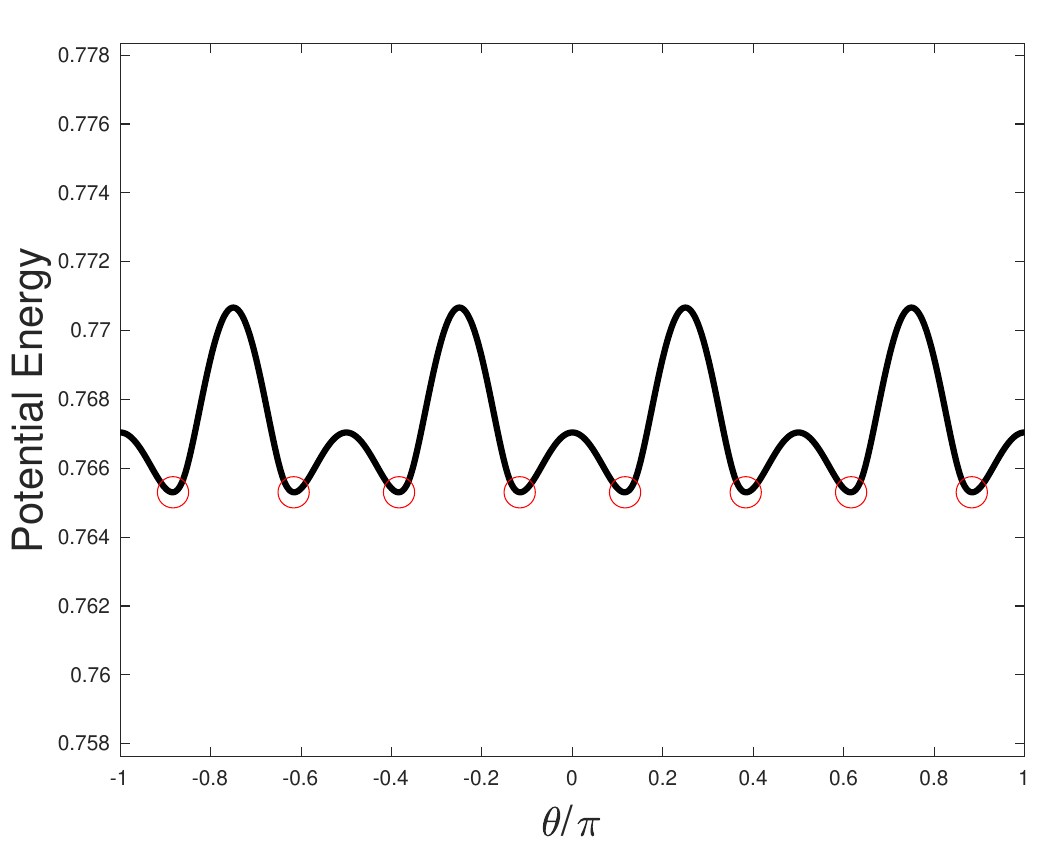} \\ 
\includegraphics[width=2.0in, angle = 0]{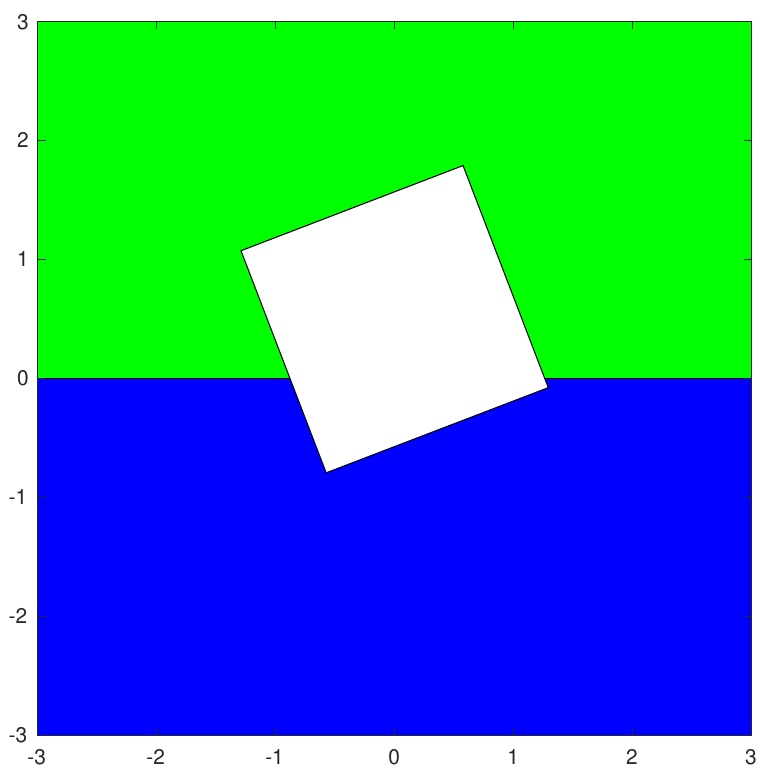}
\includegraphics[width=2.0in, angle = 0]{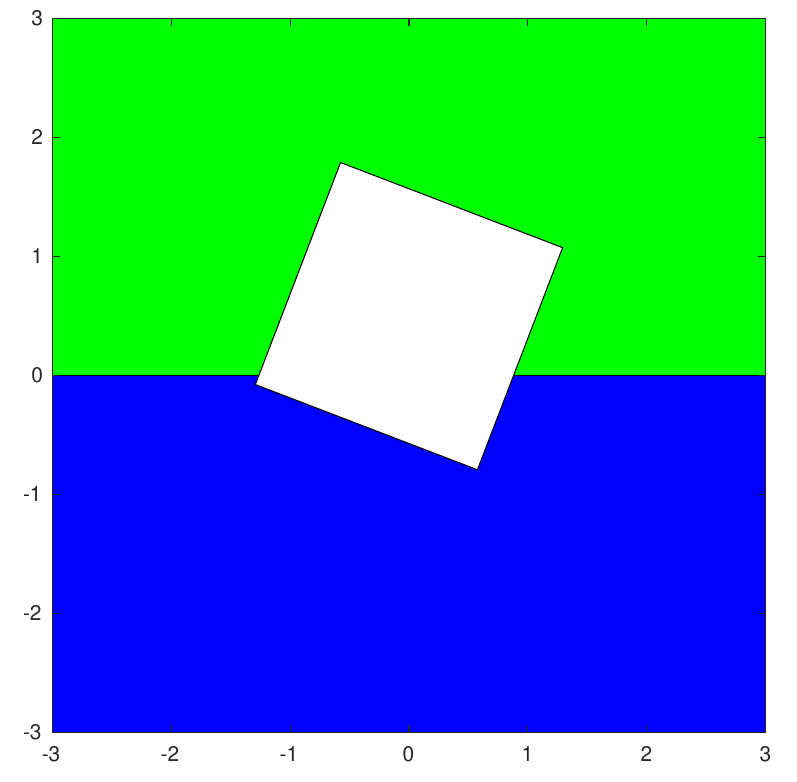} \\ 
\includegraphics[width=1.8in, angle = 0]{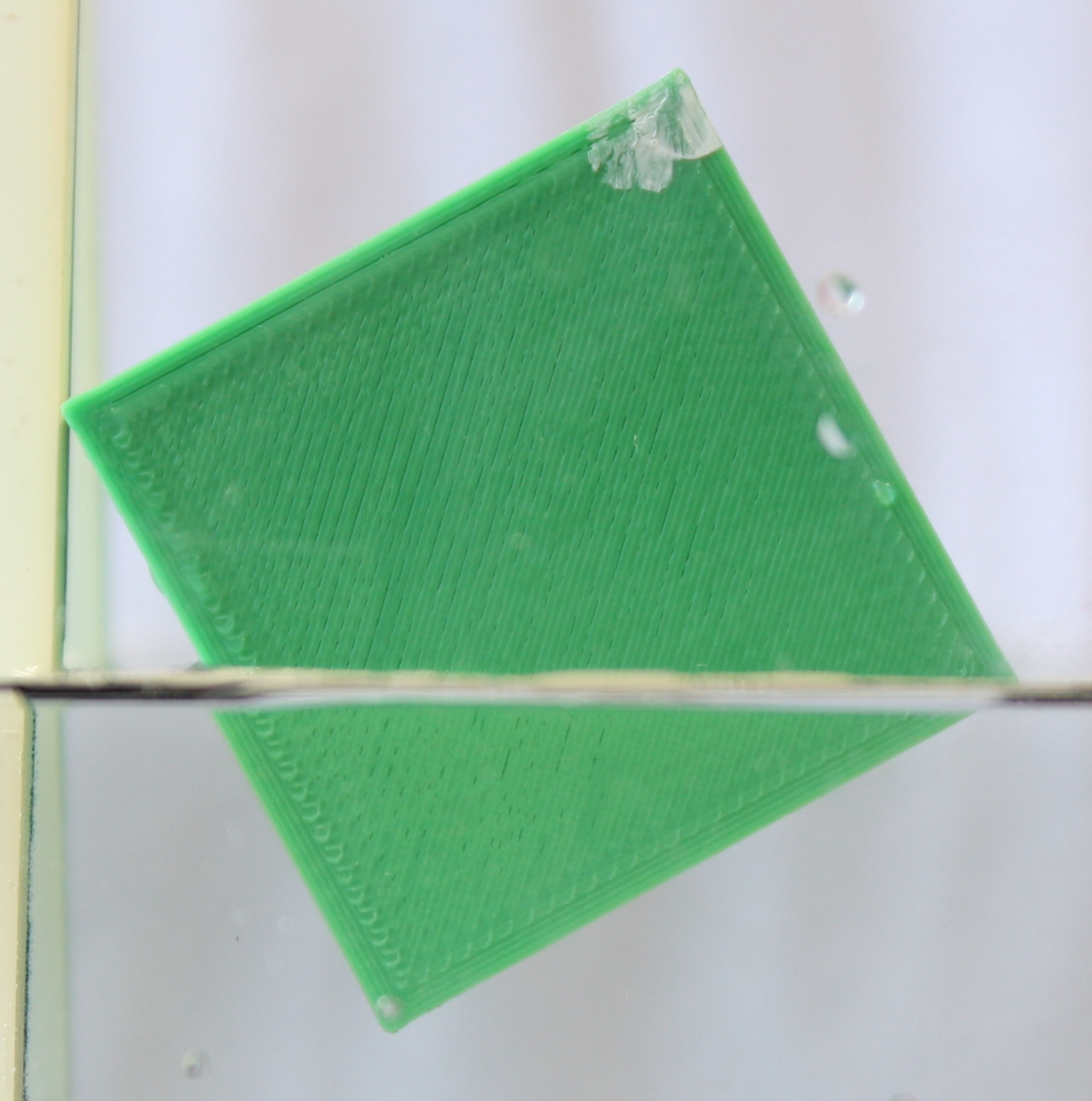} \hspace{0.2in}
\includegraphics[width=1.8in, angle = 0]{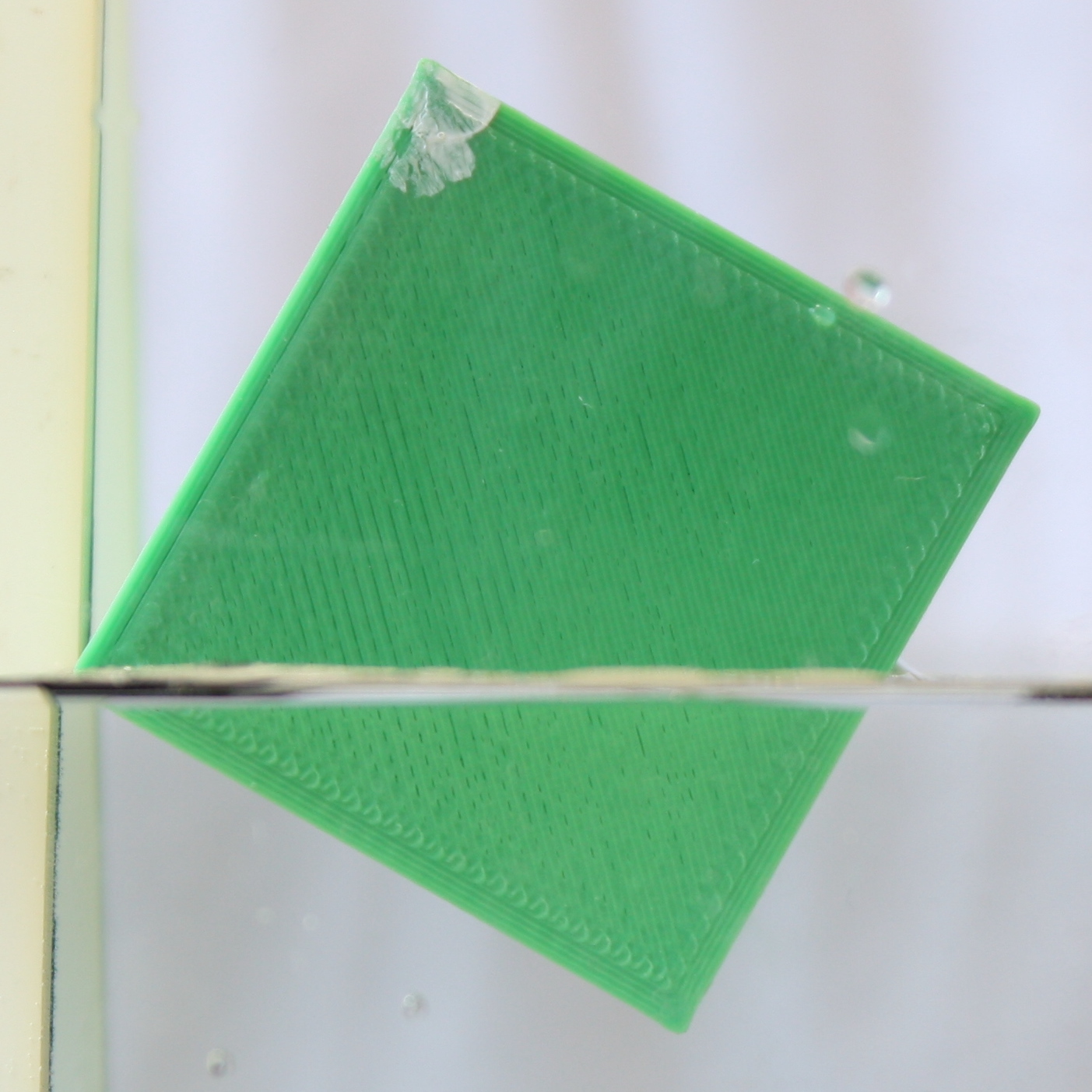}
\vspace{0.15in}
\end{center}
\caption{The upper plot shows the potential energy landscape for the square with $\vec{G}=(0,0)$ and $R=0.23296$.  There are eight stable equilibria.  
Two stable floating configuration corresponding to the orientations closest to $\theta=0$ are shown in the plots in the second row.  By symmetry these also match with the 
other stable orientations on $\theta \in [-\pi,\pi]$.  Corresponding experimental images are also shown in the bottom row.  Measured angles for this case are 
shown in Figure~\ref{fig-THETA_VS_R_zeroG}.}
\label{fig-square_PE_and_SHAPE_10}
\end{figure}

\begin{figure}[h!]
\begin{center}
\includegraphics[height=2.1in, angle = 0]{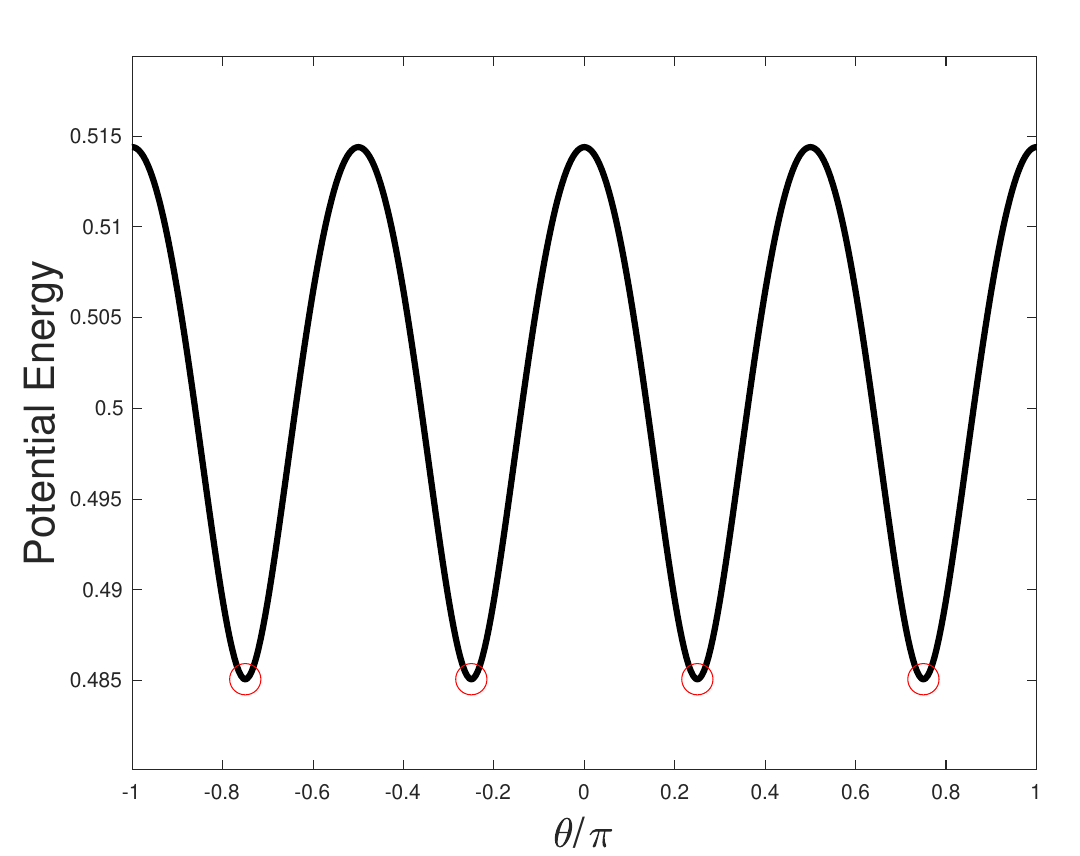}
\includegraphics[height=2.0in, angle = 0]{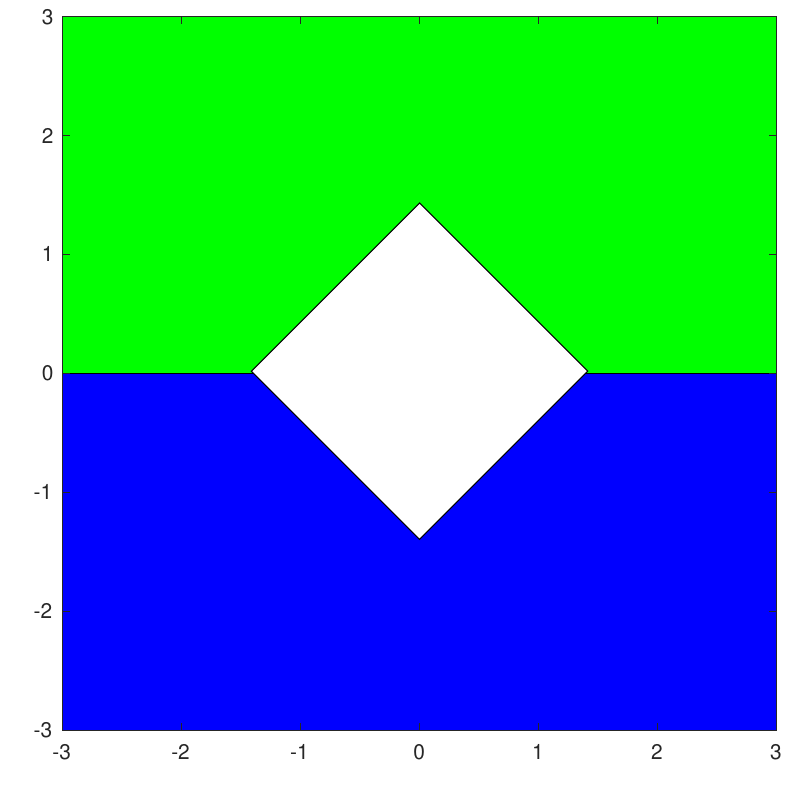}  \\ 
\includegraphics[width=1.5in, angle = 0]{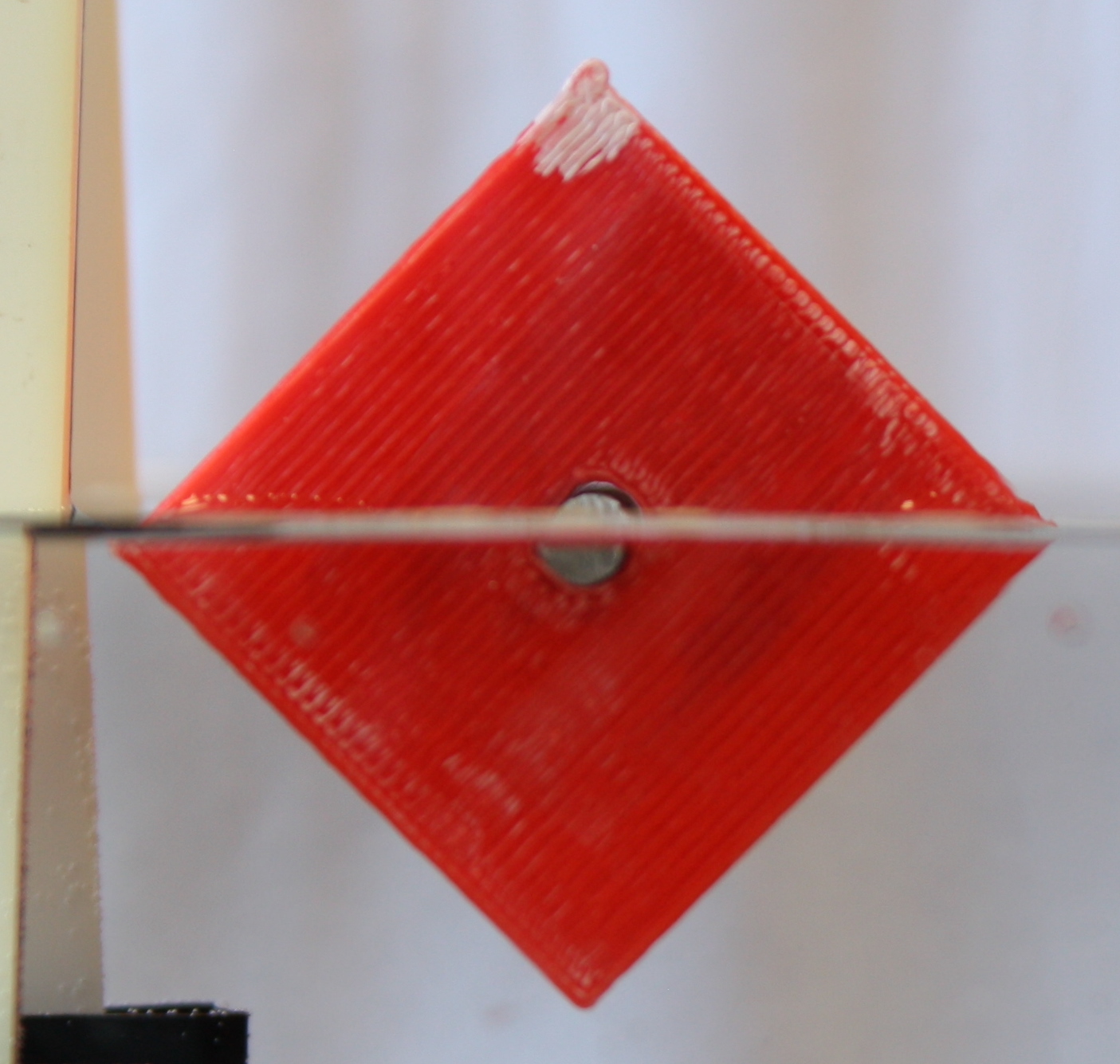} 
\vspace{0.15in}
\end{center}
\caption{The upper left plot shows the potential energy landscape for the square with $\vec{G}=(0,0)$ and $R=0.4856$.  There are four stable equilibria
corresponding to $\theta = \frac{\pi}{4} \pm \frac{\pi}{2} n$ for integer $n$.   The upper right plot shows that the square floats with vertex pointing upwards.
These four stable orientations are also observed experimentally (one such orientation is shown in the image).  Measured angles for this case are 
shown in Figure~\ref{fig-THETA_VS_R_zeroG}.}
\label{fig-square_PE_and_SHAPE_31n}
\end{figure}

\begin{figure}[h!]
\begin{center}
\includegraphics[height=2.1in, angle = 0]{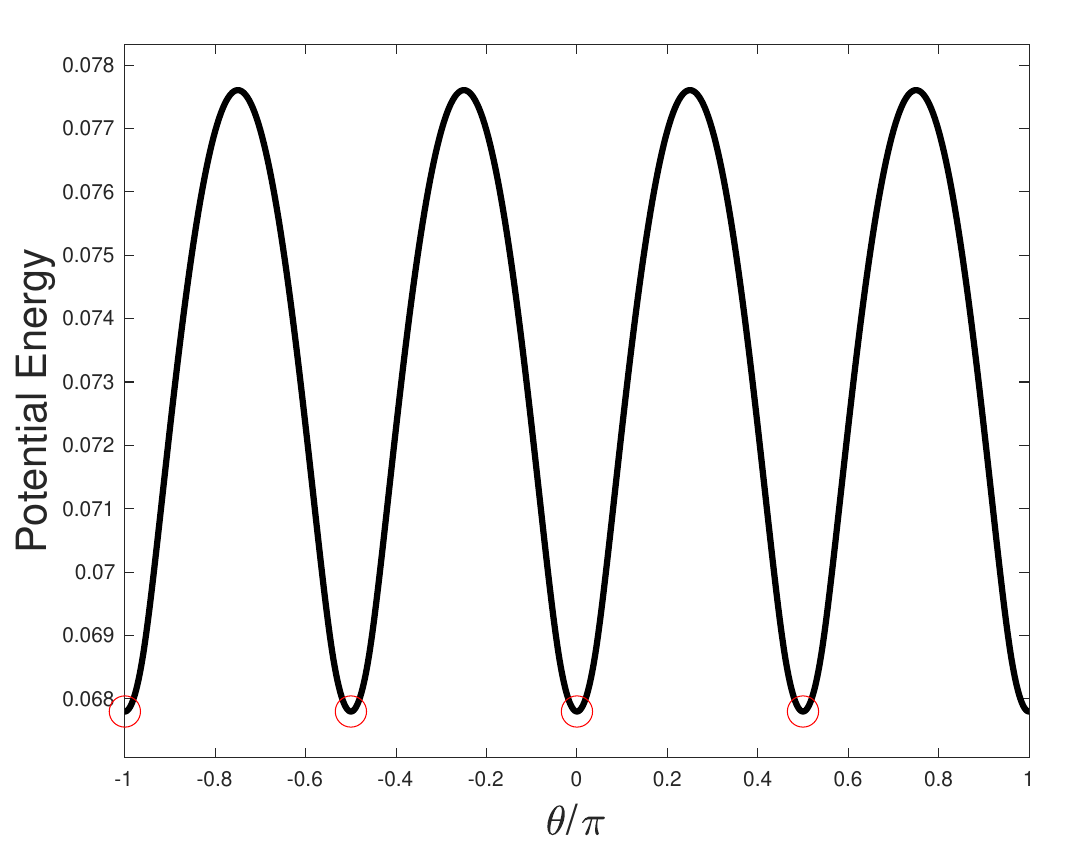}
\includegraphics[height=2.0in, angle = 0]{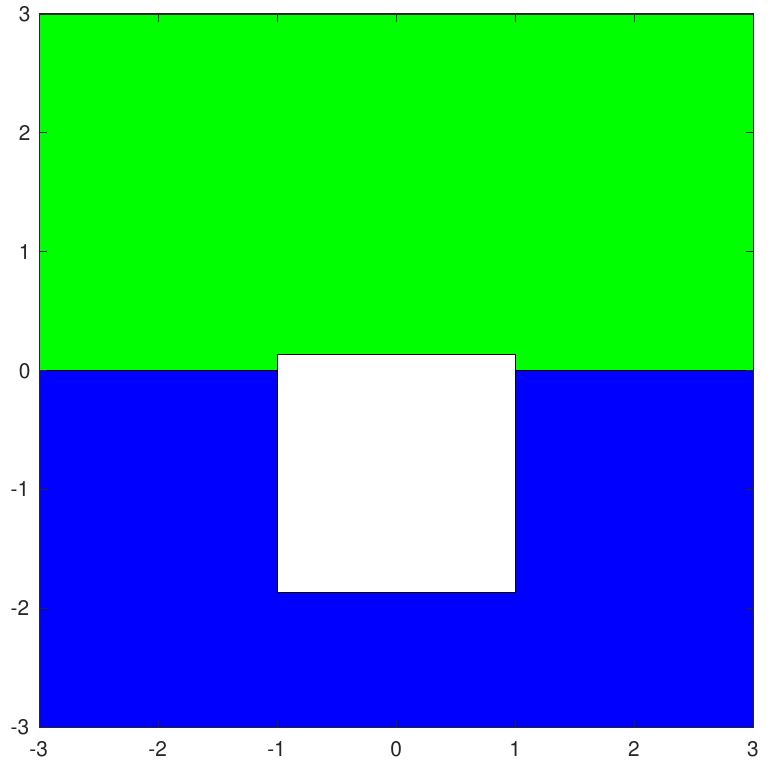}  \\ 
\includegraphics[width=1.5in, angle = 0]{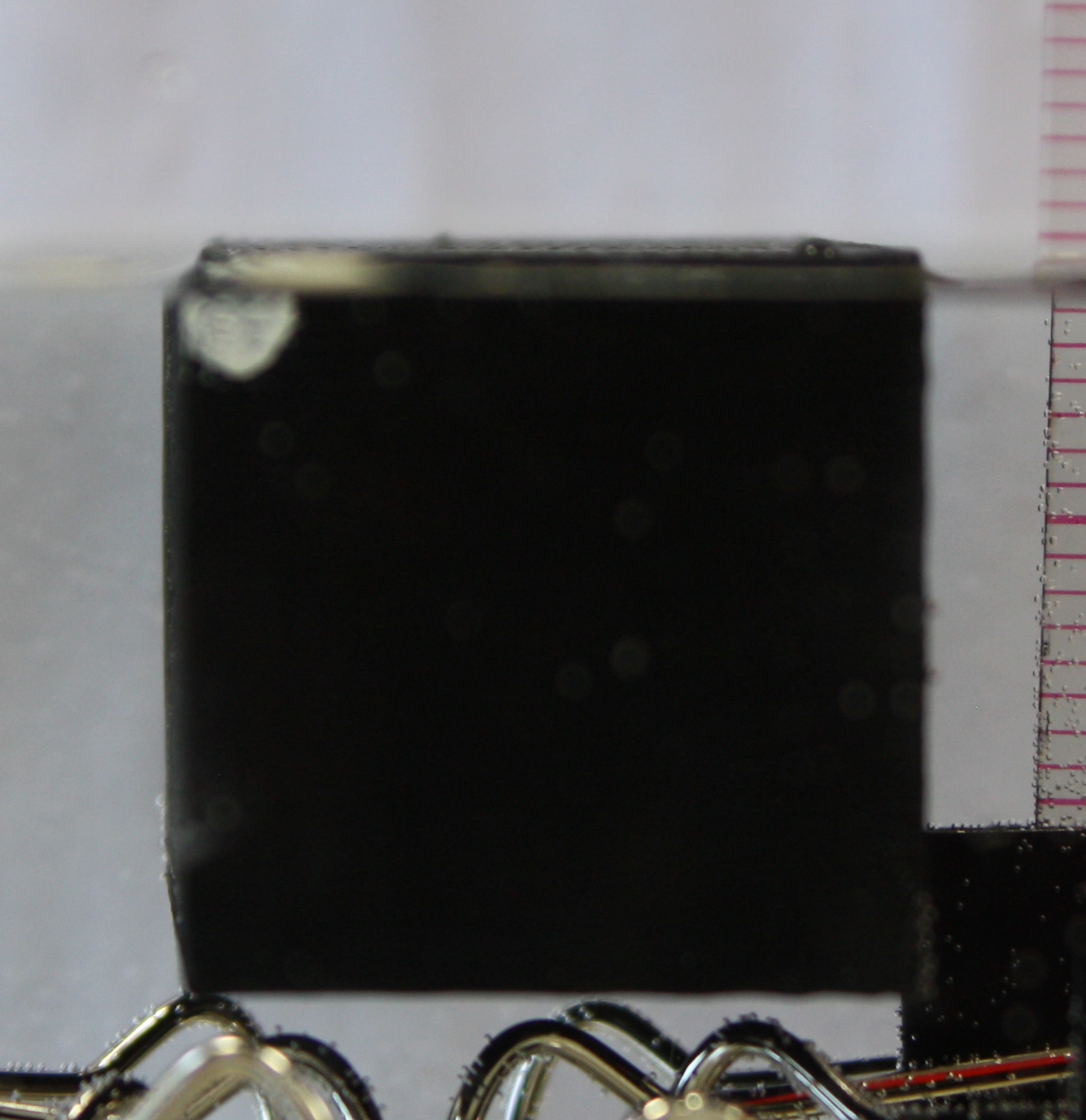} 
\vspace{0.15in}
\end{center}
\caption{The upper left plot shows the potential energy landscape for the square with $\vec{G}=(0,0)$ and $R=0.9322$.  There are four stable equilibria
corresponding to $\theta = 0 \pm \frac{\pi}{2} n$ for integer $n$.   The upper right plot shows that the square floats deep in the water with flat side up.
These four stable orientations are also observed experimentally (one such orientation is shown in the image).  Measured angles for this case are 
shown in Figure~\ref{fig-THETA_VS_R_zeroG}.}
\label{fig-square_PE_and_SHAPE_56}
\end{figure}

\subsection{Floating Squares: Breaking Symmetry}

Figure~\ref{fig-PEplots31-32-33} show the potential energy landscape for three
different prints corresponding to square cross sections with nail-filled holes at {\bf A}: $(0,0)$ (upper left plot),
{\bf B:} $(0.3,0.3)$ (upper right plot), and {\bf C}: $(0.45,0.45)$ (lower plot).  
Note that the coordinates for the holes are given in units of $s/2$ where $s$ is the length of
the side of the square.    These three have density ratios 
of $R=0.4856$ for Case {\bf A},
$R=0.4874$ for Case {\bf B},
and $R=0.4911$ for Case {\bf C}. 

For Case {\bf A} with nail-filled hole at $(0,0)$ already discussed in 
Figures~\ref{fig-THETA_VS_R_zeroG} and~\ref{fig-square_PE_and_SHAPE_31n}, 
four stable orientations, with a corner of the square pointing straight up, are predicted
theoretically and observed experimentally. 

Experiments for Case {\bf B}, with a nail-filled hole at $(0.3,0.3)$,
are shown in Figure~\ref{fig-EXP_plots_32nail_cropped}.
In this case the center of gravity $\vec{G}$ is no longer at the center of the square and the
symmetry is broken. Despite this broken symmetry, four stable orientations are still observed experimentally
as shown in Figure~\ref{fig-EXP_plots_32nail_cropped}. Our
theory predicts that $(G_x,G_y) = (0.06789,0.06789)$ and the corresponding
potential energy plot is shown by the solid curve in the upper right plot of
Figure~\ref{fig-PEplots31-32-33}.  There are only two stable
orientations predicted at this value of $\vec{G}$ and so our theory
does not match the experimental observations.  However, two other curves are shown in the
upper right plot in Figure~\ref{fig-PEplots31-32-33} -- the dashed
curve has $(G_x,G_y) = (0.05,0.05)$ and the dash-dotted curve has $(G_x,G_y) =
(0.04,0.04)$.  These curves with nearby values of $\vec{G}$ indicate that there are two other stable orientations nearby.
One explanation for the discrepancy between experiment and theory is that the center of gravity of our print is not exactly where
we predict it to be, perhaps due to uncertainties in the infill structure of the print. 
Another potential source of imprecision in the center of gravity is the nail not fitting 
precisely in the middle of the hole, but the nail seems to fit snugly in the hole, so it seems a less likely explanation. 
Also, there is evidence of menisci at the solid--liquid--air contact line in Figure~\ref{fig-EXP_plots_32nail_cropped} suggesting that surface tension could provide
a large enough force to hold the print in an otherwise slightly unstable configuration.

For Case {\bf C}, with a nail-filled hole at $(0.45,0.45)$ and $(G_x,G_y)=(0.1017,0.1017)$ the potential energy landscape shown in the lower plot of Figure~\ref{fig-PEplots31-32-33} 
indicates that only two stable equilibria exist.  For this case both theory and experiment are
in agreement on the number of stable equilibria.  
The experimental images for this case are not shown but are very similar to the upper left and lower left images of
Figure~\ref{fig-EXP_plots_32nail_cropped}.  Corresponding stable orientations with the nail to the left or right as in the upper right and lower right
images of Figure~\ref{fig-EXP_plots_32nail_cropped} no longer exist for Case {\bf C}.


\begin{figure}[h!]
\begin{center}
\includegraphics[width=2.25in, angle = 0]{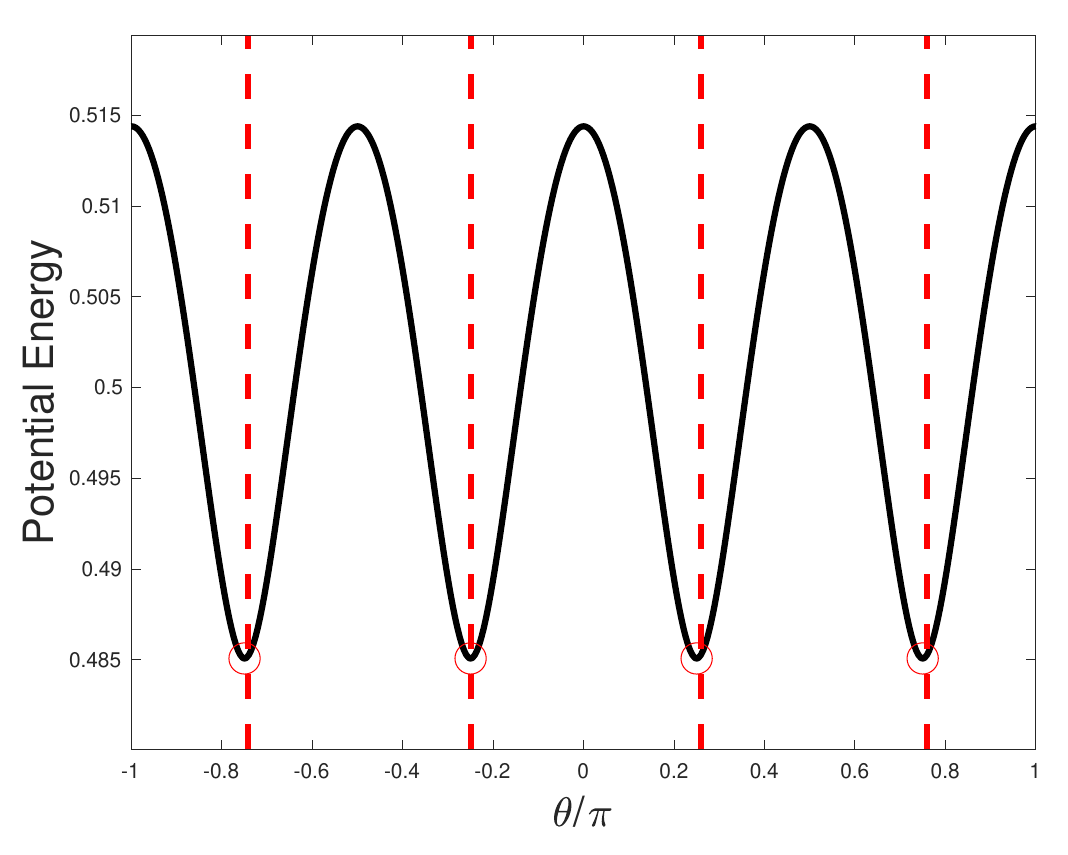} 
\includegraphics[width=2.25in, angle = 0]{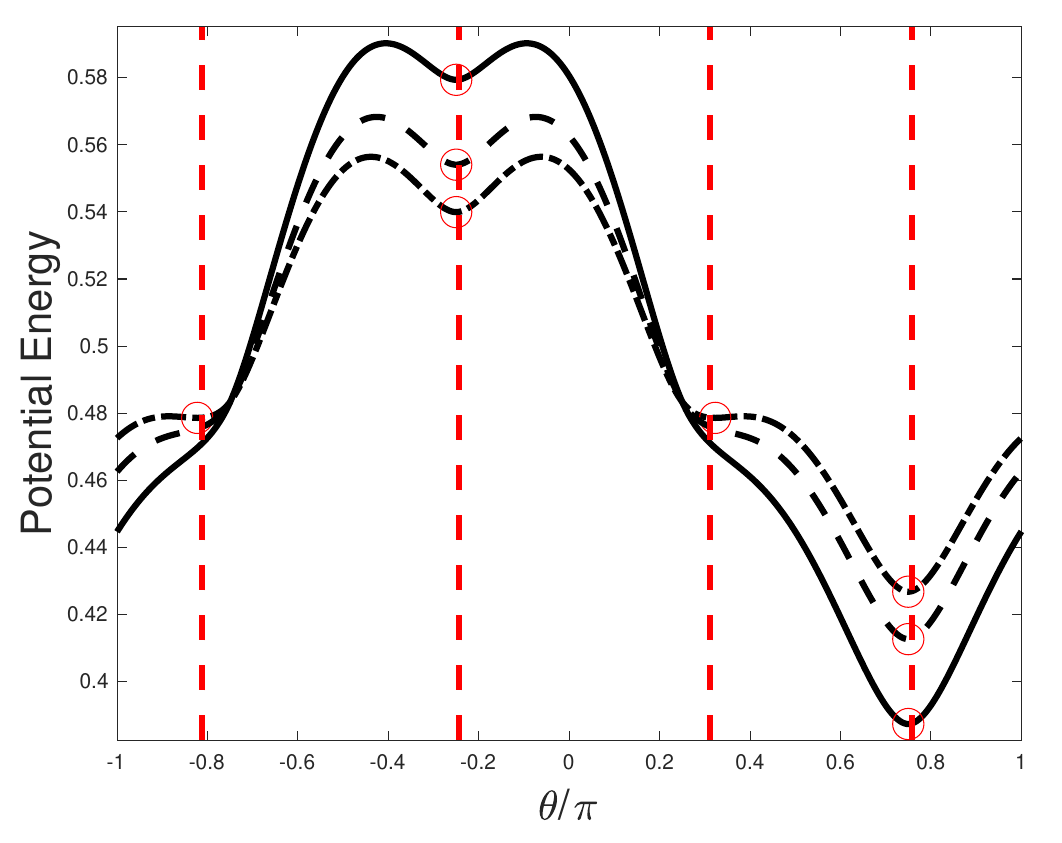}  \\ 
\includegraphics[width=2.25in, angle = 0]{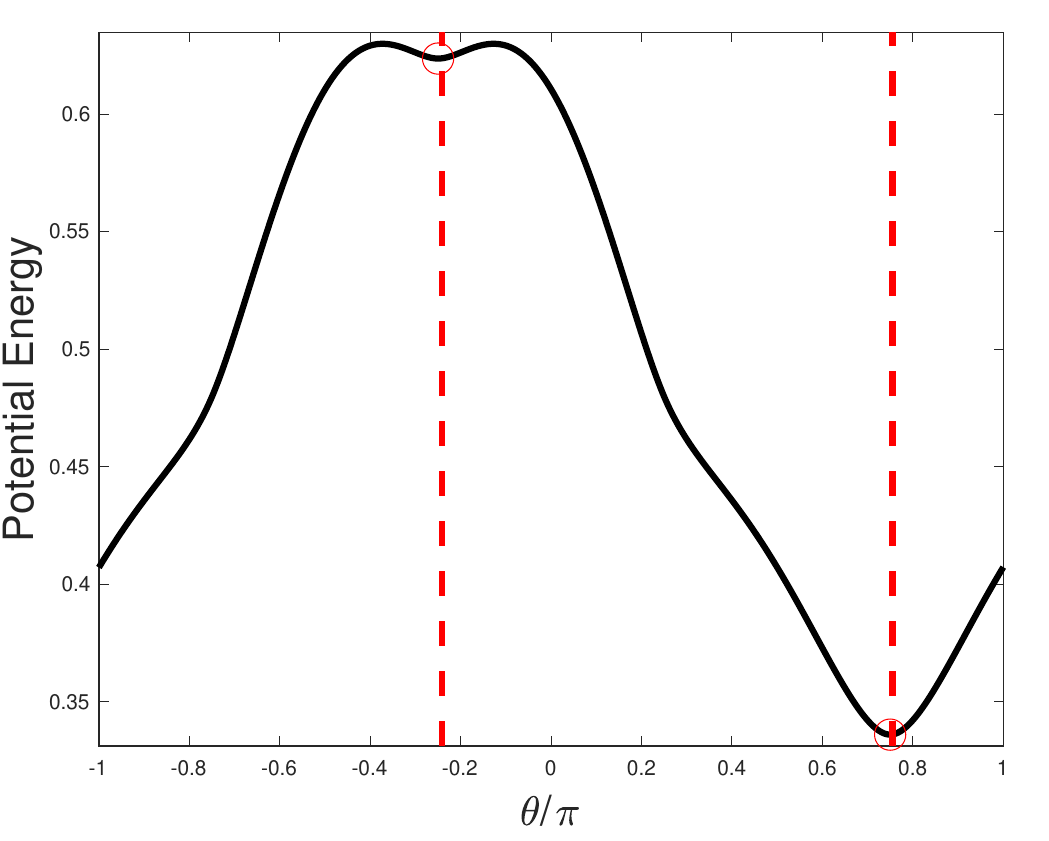} \\ 
\end{center}
\caption{Potential energy plots for squares with nail-filled holes at {\bf A}: $(0,0)$ (upper left plot),
{\bf B:} $(0.3,0.3)$ (upper right plot), and {\bf C}: $(0.45,0.45)$ (lower plot).
The black curves (solid, dashed, or dash-dotted) show the potential energy function defined in equation~(\ref{eq:PE_formula_SQUARE}).  
The open red circles indicate the theoretical local minima of the potential energy.  The vertical red dashed lines show the angles at which experimentally-floating squares appear to be stable.    
Case {\bf A} has four-fold symmetry and four stable orientations are predicted theoretically and observed experimentally (see
also Figure~\ref{fig-square_PE_and_SHAPE_31n}).   
For Case {\bf B}, in which the symmetry is broken, our theory predicts only two stable orientations but we observe four experimentally (these four orientations are
shown in Figure~\ref{fig-EXP_plots_32nail_cropped}).    The dashed and dash-dotted black curves in the upper right plot show two other potential energy landscapes
for nearby values of $\vec{G}$ (see text for details) indicating the presence of nearby stable states.
For Case {\bf C}, the square is farther from symmetric and both theory and experiment predict only two stable
orientations.   We do not show experimental images for Case {\bf C}, but they are similar to the upper left and lower left images in Figure~\ref{fig-EXP_plots_32nail_cropped}.}
\label{fig-PEplots31-32-33}
\end{figure}

\begin{figure}[tb]
\includegraphics[width=0.4\textwidth]{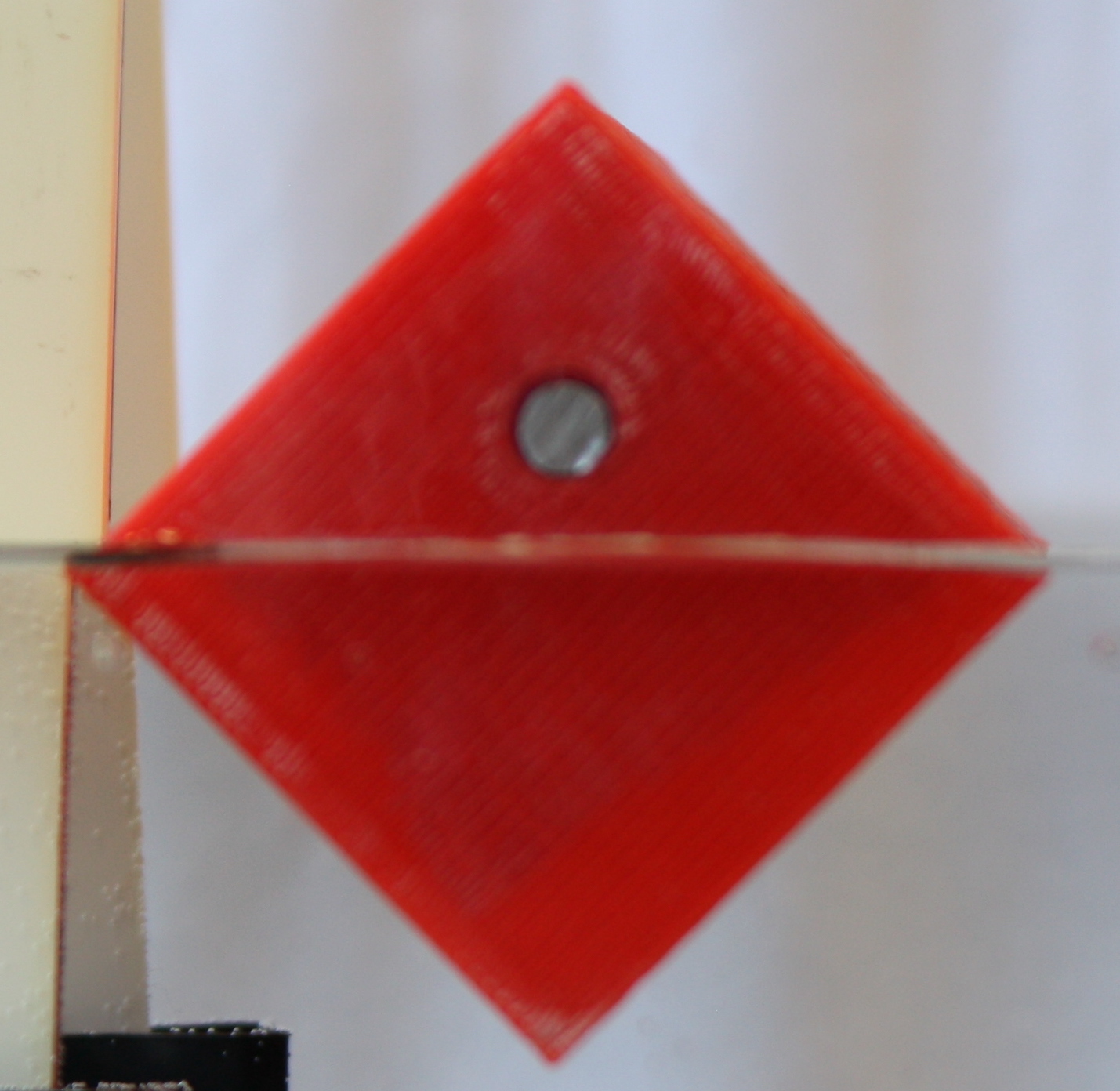}
\includegraphics[width=0.4\textwidth]{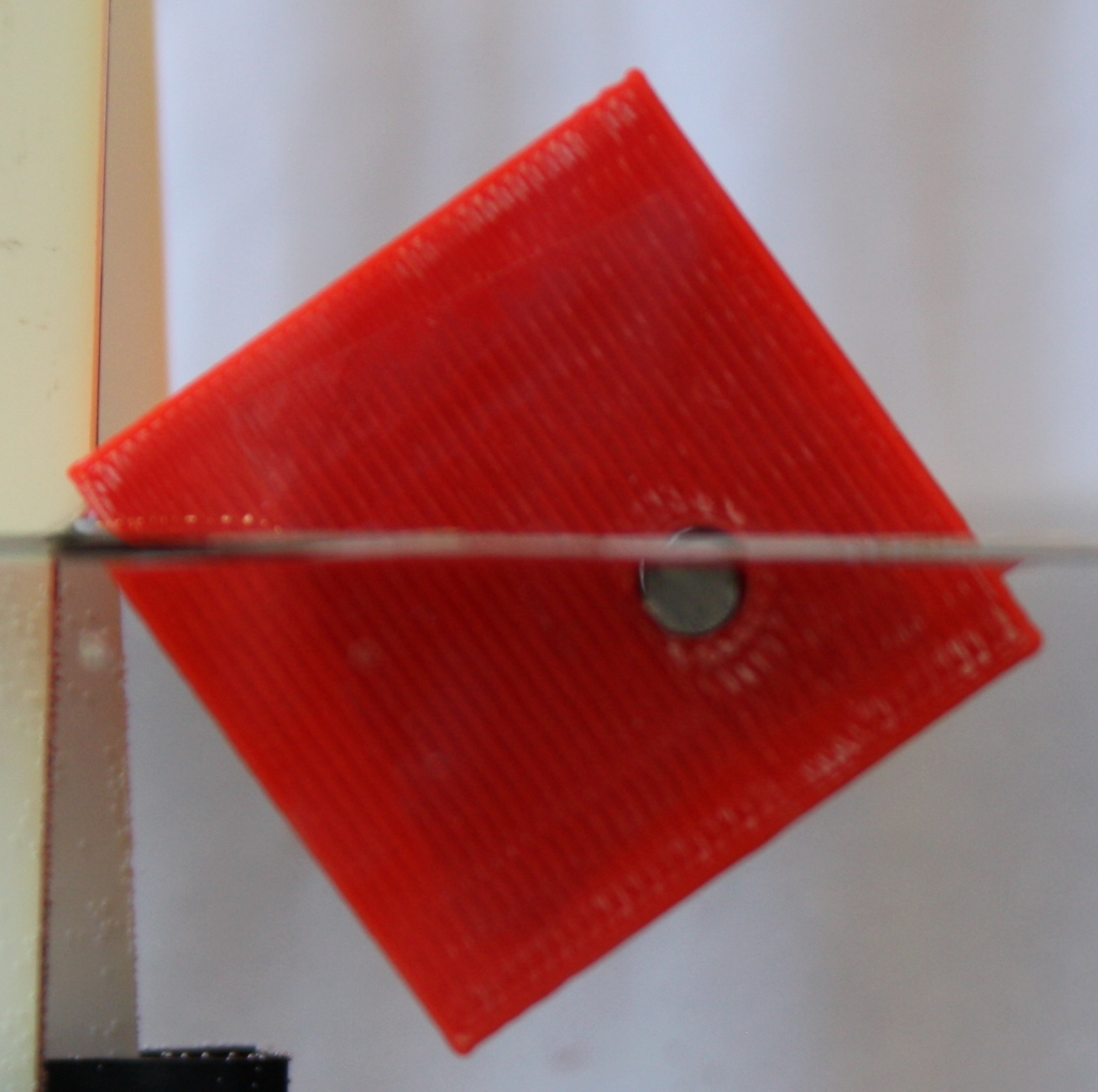} \\
\includegraphics[width=0.4\textwidth]{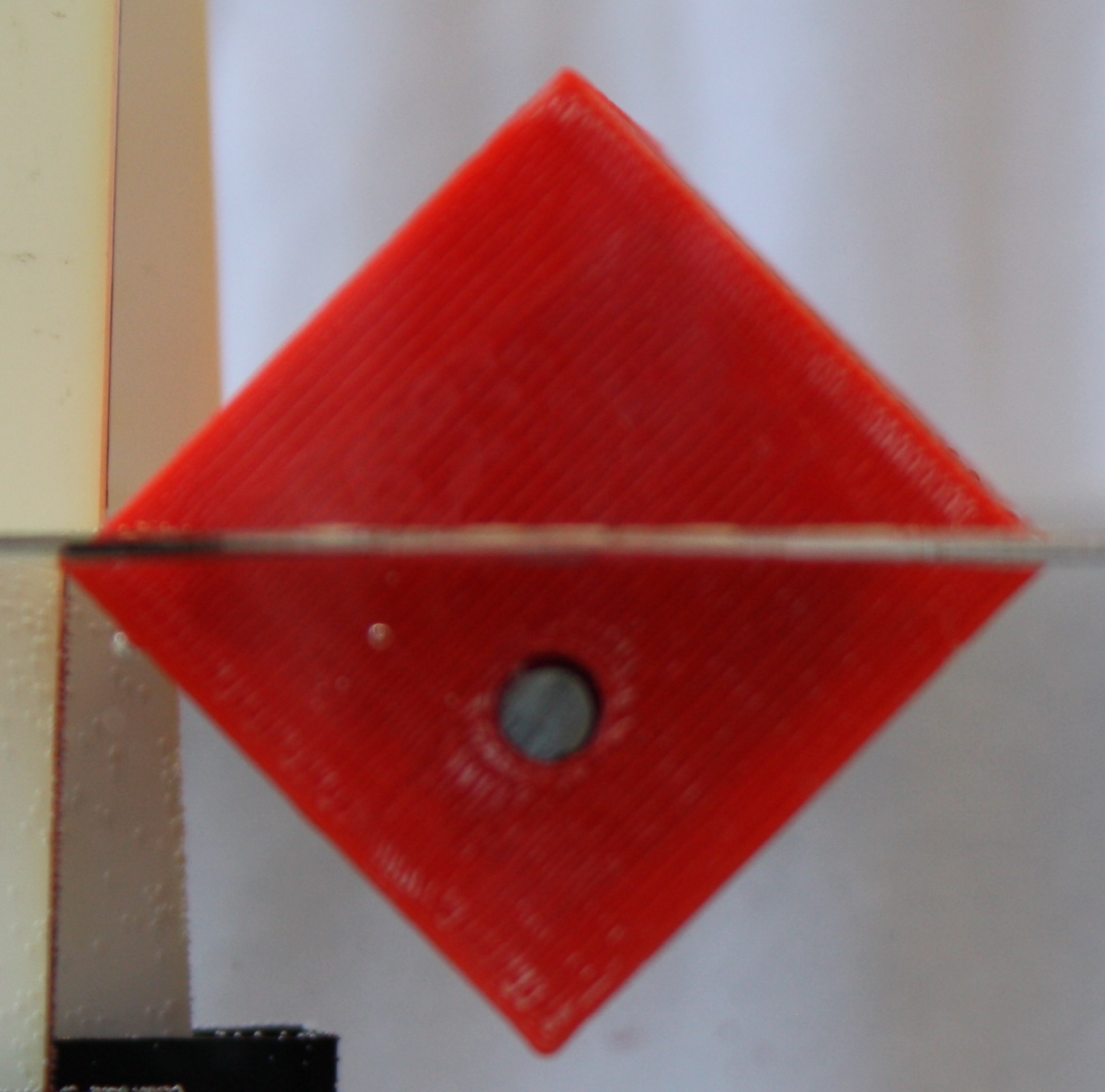}
\includegraphics[width=0.4\textwidth]{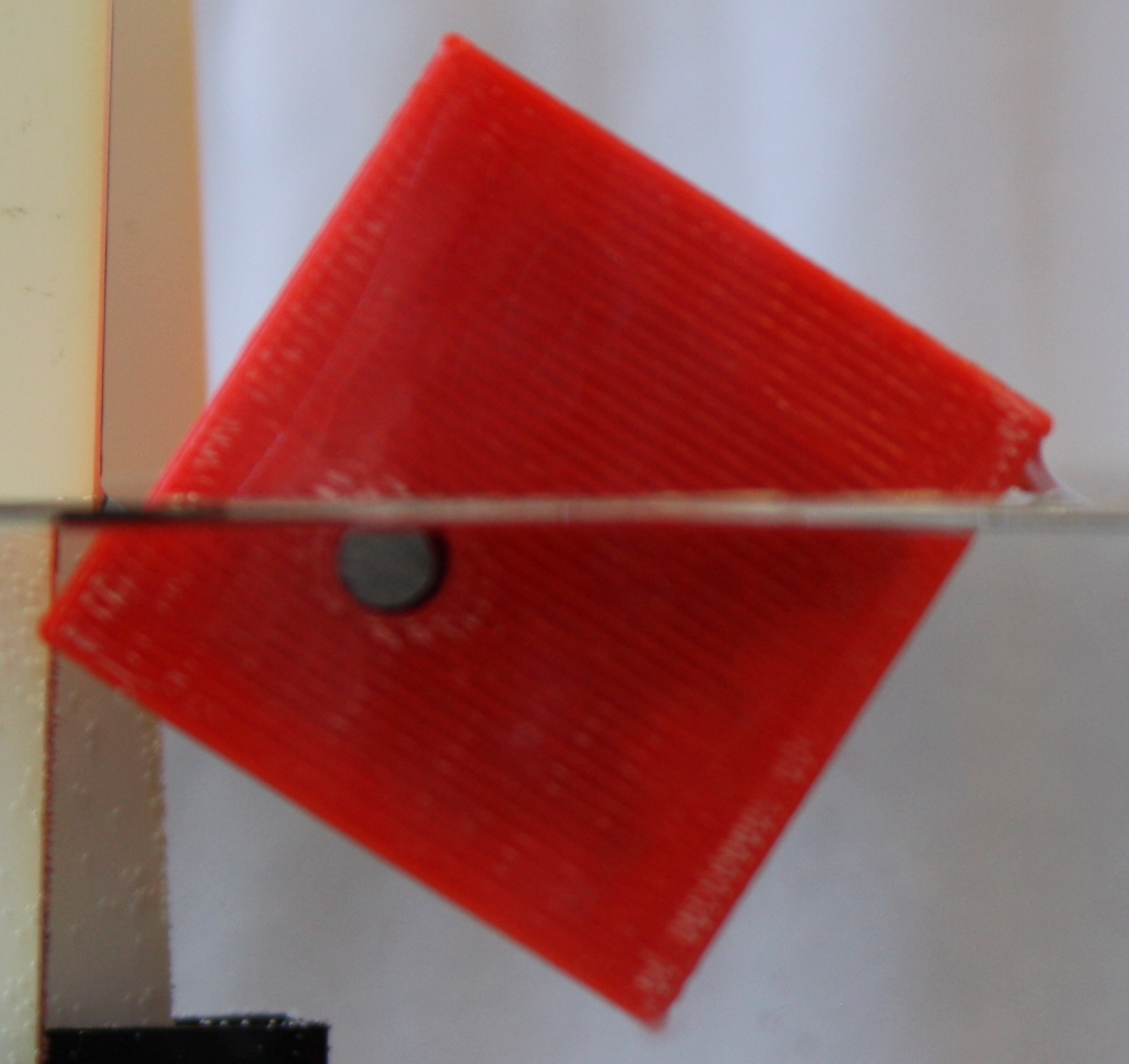}
\caption{Four floating orientations for Case {\bf B} in Figure~\ref{fig-PEplots31-32-33} with a nail-filled hole at $(0.3,0.3)$ giving a nonzero center of gravity $\vec{G}$.  
When the hole is {\em up} or {\em down} the stable
orientations correspond to corner straight up.  When the hole is {\em left}
or {\em right} the top corner of the square is deflected counterclockwise or
clockwise, respectively.}
\label{fig-EXP_plots_32nail_cropped}
\end{figure}

As we have just observed, when $\vec{G} \neq 0$ the symmetry of the square is
broken and the predicted number of stable equilibria change.  We display this in
Figure~\ref{fig-THETA_vs_R_32nail} for four cases $\vec{G} = (0.01,0.01)$,
$(0.02,0.02)$, $(0.04,0.04)$, and $(0.06789,0.06789)$.  On the last two of these
we have plotted the four experimentally-observed stable orientation angles for the print denoted
Case {\bf B} with a nail-filled hole at $(0.3,0.3)$.  As just discussed for Case {\bf B}, 
our theory predicts $\vec{G}=(0.06789,0.06789)$ and only two equilibrium values
of $\theta$ exist for this case as shown in the lower right plot of Figure~\ref{fig-THETA_vs_R_32nail}.
Reduction in the value of $\vec{G}$ to $(0.04,0.04)$, for example as indicated in the lower
left plot of Figure~\ref{fig-THETA_vs_R_32nail} shows four stable orientations at
the density ratio, $R=0.4874$, of Case {\bf B}.  The upper two plots corresponding
to $\vec{G} = (0.01,0.01)$ and $(0.02,0.02)$.  Especially in comparison to 
the predicted stable angles shown in Figure~\ref{fig-THETA_VS_R_zeroG} these plots reveal
interesting bifurcation structure as the square symmetry is broken.  Note that 
in addition to floating squares with four or eight stable orientations, when the symmetry
is broken in this way (moving the center of gravity towards a corner) there are also 
situations where either three or six stable orientations are predicted.

\begin{figure}[tb]
\includegraphics[width=0.45\textwidth]{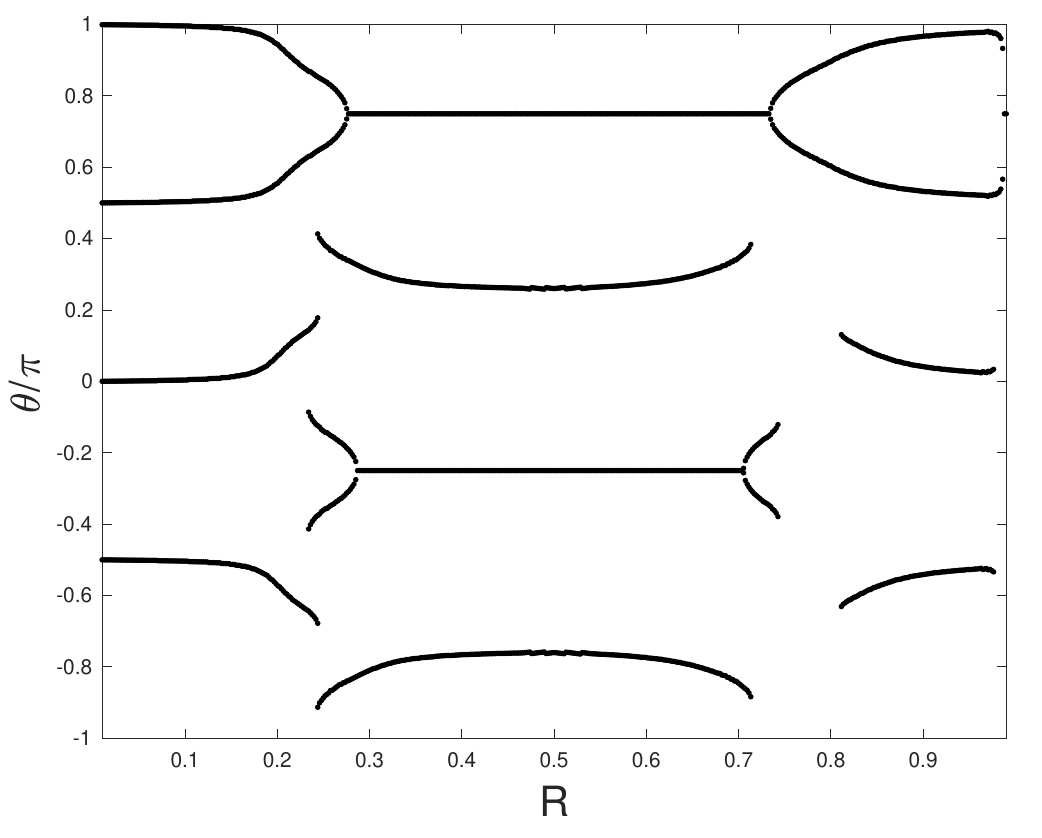}
\includegraphics[width=0.45\textwidth]{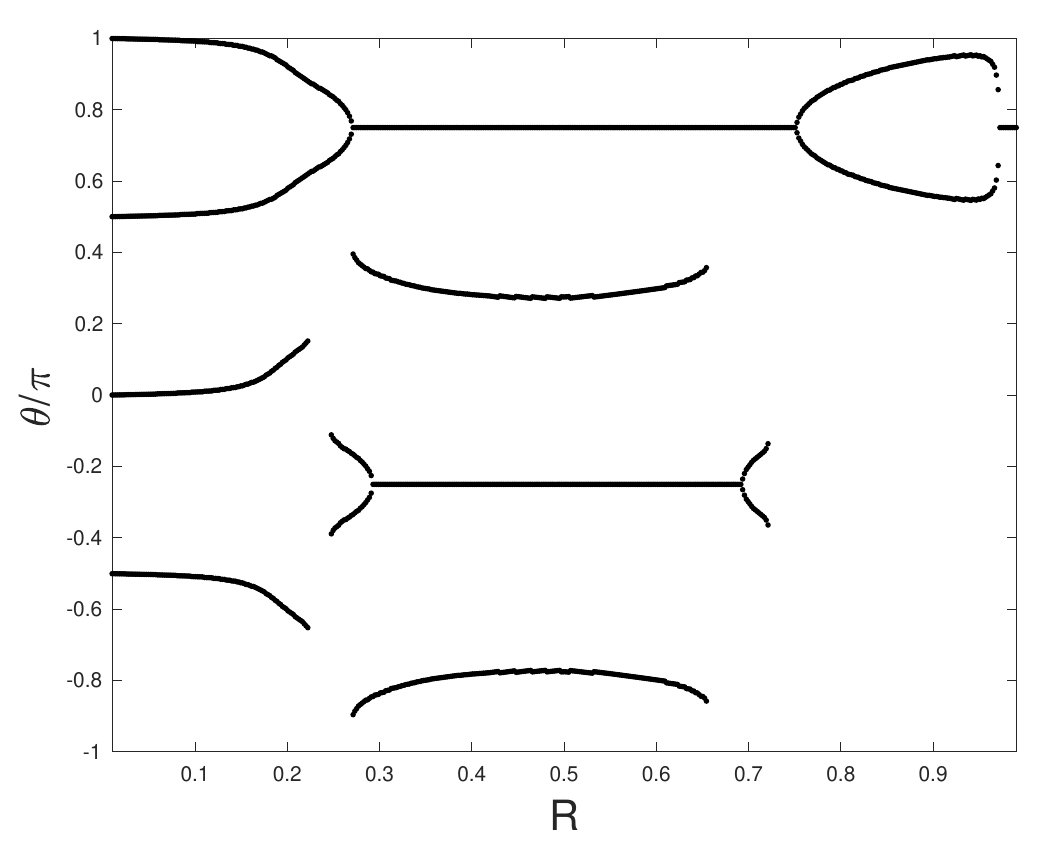} \\ 
\includegraphics[width=0.45\textwidth]{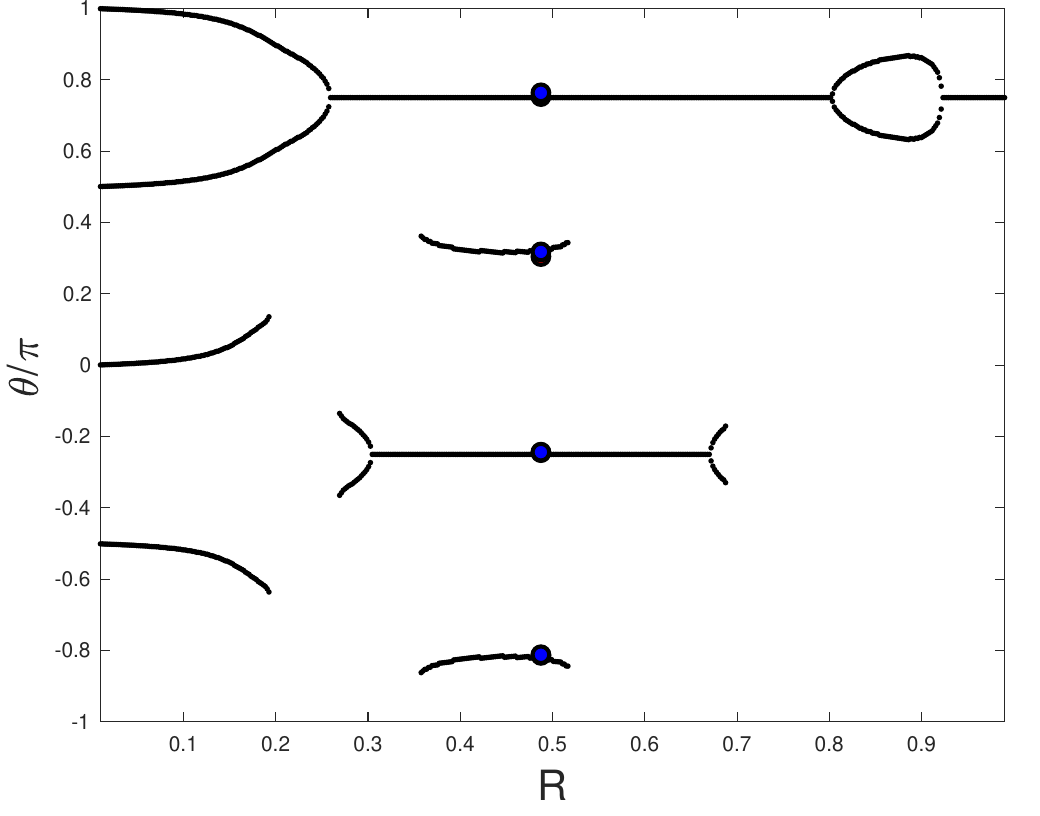}
\includegraphics[width=0.45\textwidth]{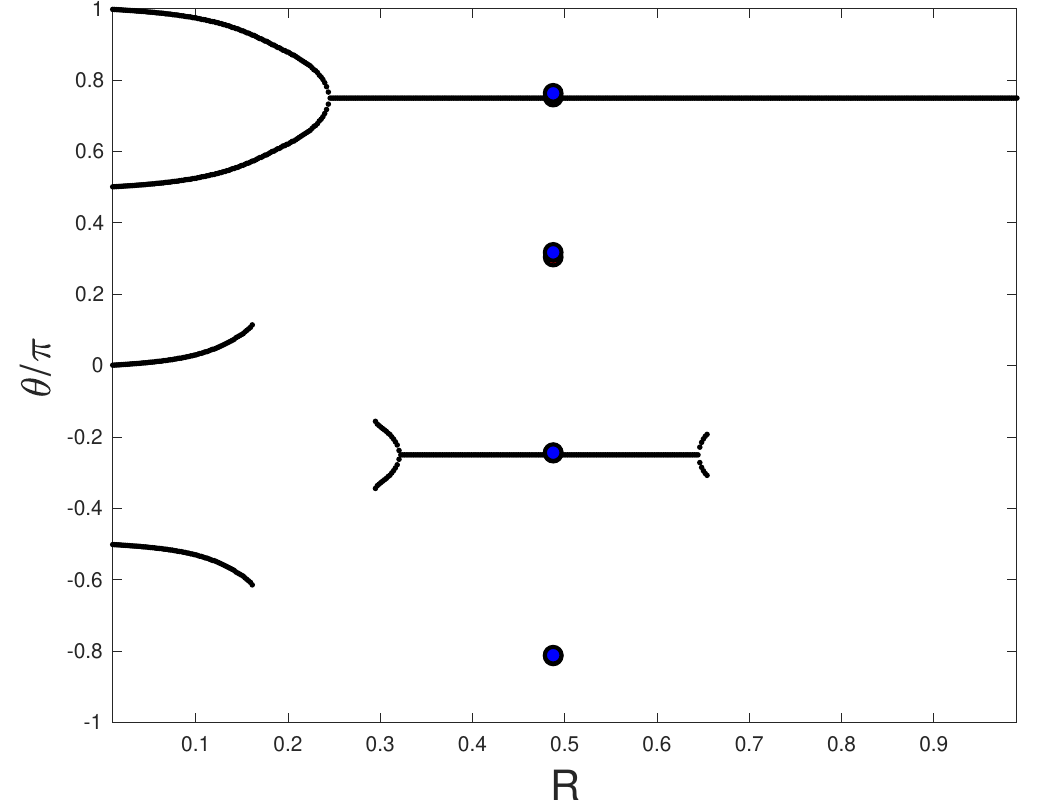} \\ 
\caption{Stable equilibrium angles versus density ratio $R$ for off-center squares.    These plots have $\vec{G} = (0.01,0.01)$ (upper left), $(0.02,0.02)$ (upper right), 
$(0.04,0.04)$ (lower left), and $(0.06789,0.06789)$ (lower right).  
Experimentally-measured equilibrium angles for Case {\bf B} corresponding to $R=0.4874$,
described in the text and in Figures~\ref{fig-PEplots31-32-33} and~\ref{fig-EXP_plots_32nail_cropped}, 
are shown in both the lower left and lower right plots.  
The full sequence of plots in this figure should be compared to the one for the symmetric square in Figure~\ref{fig-THETA_VS_R_zeroG} which has $\vec{G}=(0,0)$.}
\label{fig-THETA_vs_R_32nail}
\end{figure}


\subsection{Results for General Polygonal Cross Sections: The Mason M}

As an example of a floating shape with nontrivial cross sectional area we chose the Mason M.
This print was made based on a counter-clockwise-oriented set of points describing the shape of the M.

The cross-sectional area of the Mason M, $A_M$, was obtained using Matlab's {\tt
polyarea.m} applied to the point set described above. This value adjusted by a pixel
to millimeter (mm) conversion factor used to scale for 3D printing gave the
area. In particular, we found that $A_M = 126,550 \cdot (0.06)^2 \mbox{ mm}^2 = 455.58$
mm$^2$.  The length of our Mason M was $70$ mm, giving a volume of $31,891$
mm$^{3}$. Since its mass was $27.92$g its corresponding effective density  is
$\rho_{\rm eff} = 0.8755$ g cm$^{-3}$, which, in comparison to the density of water, is fairly close to the typical density ratio of an iceberg in a
polar sea.
Figure~\ref{mason} shows the potential energy landscape along with predicted stable configurations 
for the Mason M. 
The four stable orientations of the floating Mason M are shown in  Figure~\ref{masonfloat}. 

%
\begin{figure}[tb]
\includegraphics[width=0.8\textwidth]{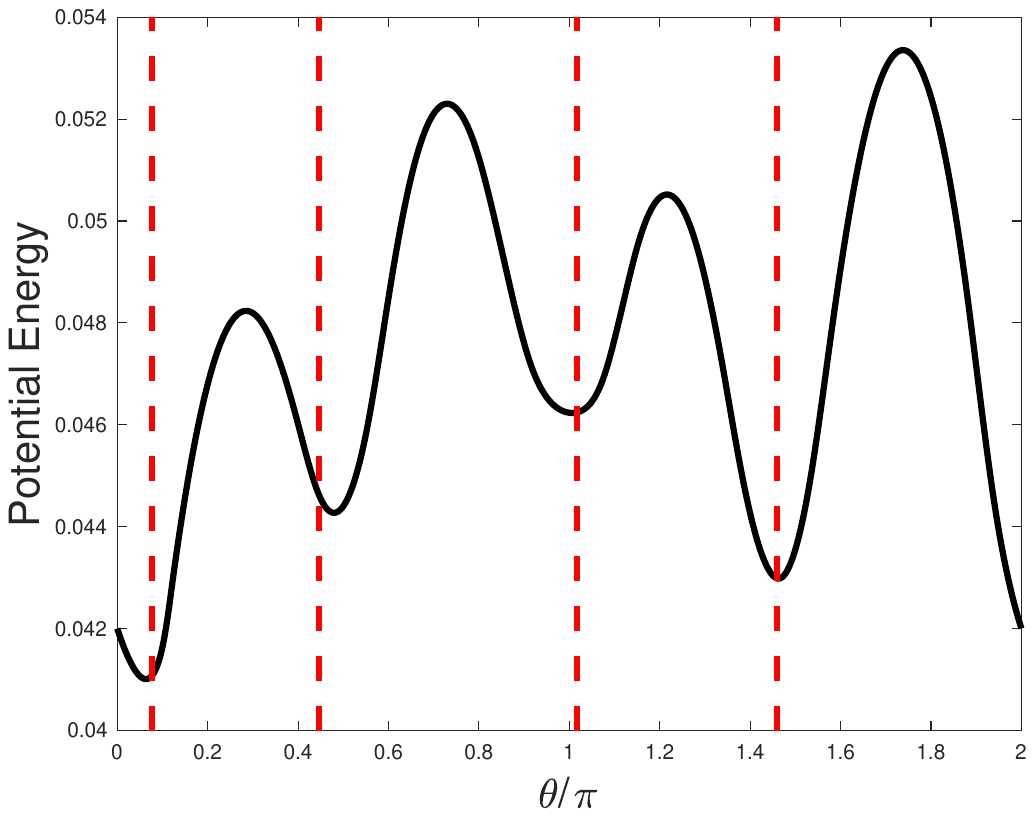}\vspace{-0.95in}
\caption{The Mason M potential energy plot where potential energy depends on the orientation angle of the waterline with respect to the 3D Mason M Print (zero angle corresponds
to an M in its usual upright orientation). Each vertical red dotted line corresponds to an experimentally found stable orientation. The local minima of the graph define stable orientations theoretically calculated using our code.
See also Figure~\ref{masonfloat}.}
\label{mason}
\end{figure}
\begin{figure}[tb]
\includegraphics[width=0.4\textwidth]{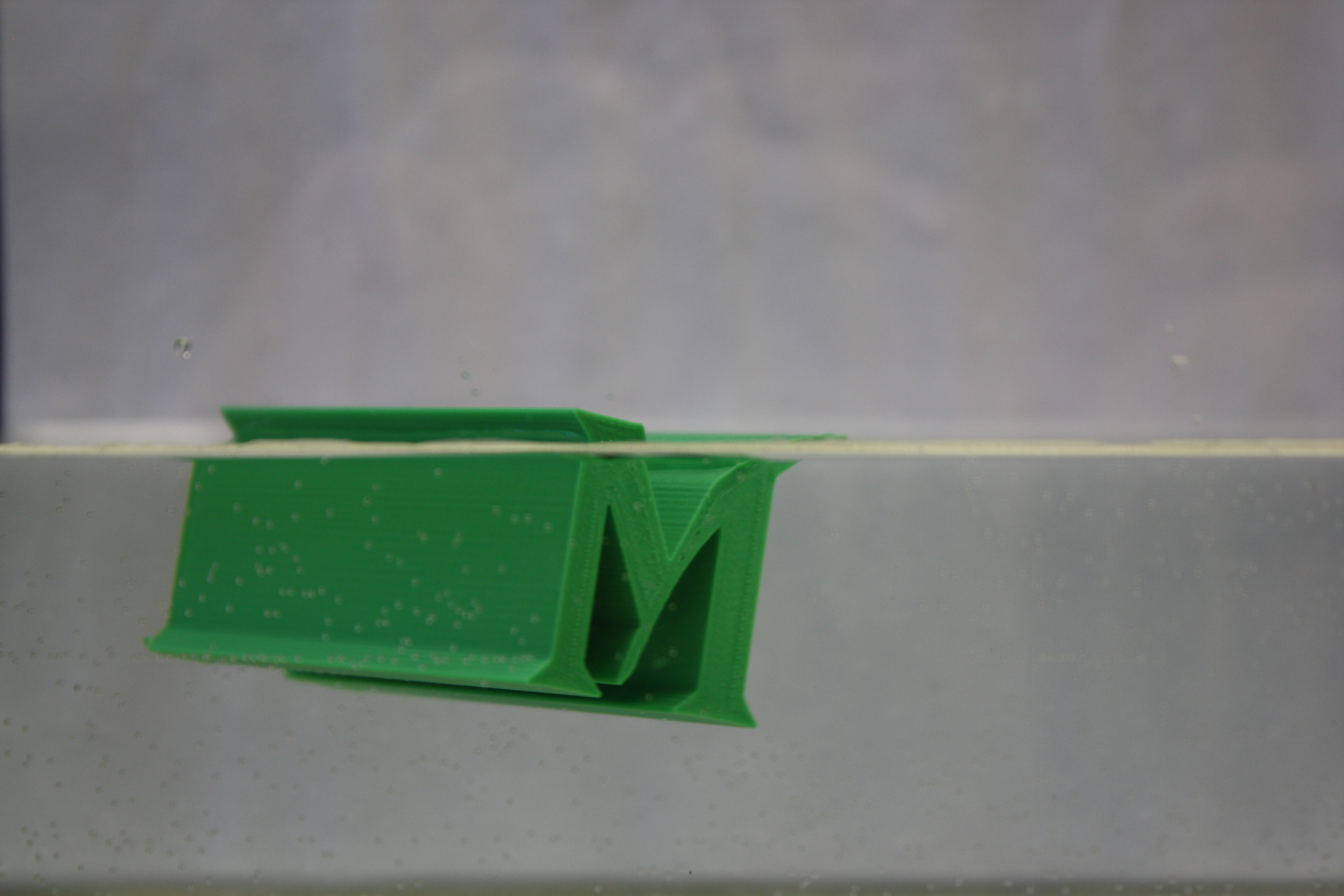}
\includegraphics[width=0.4\textwidth]{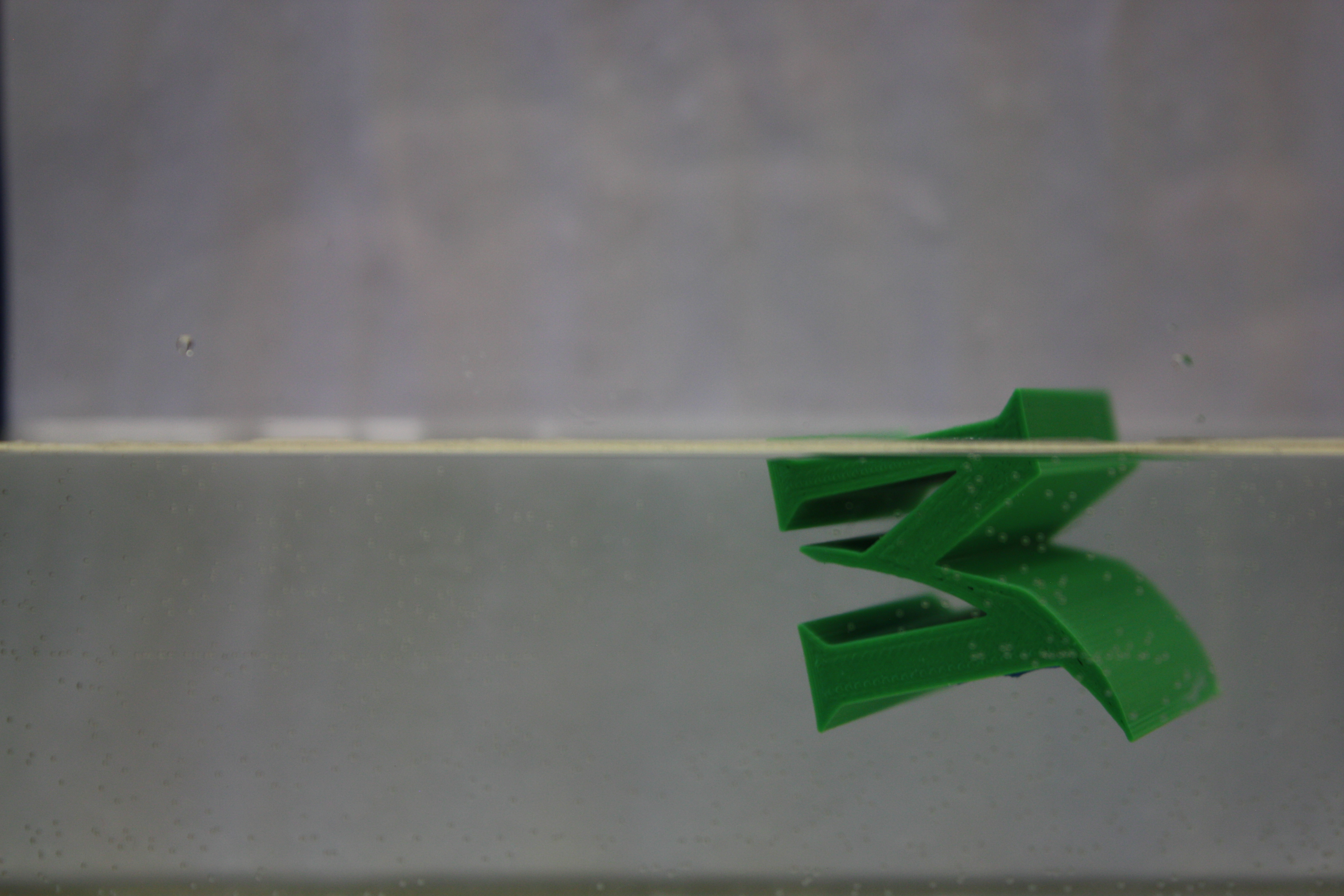} \\
\includegraphics[width=0.4\textwidth]{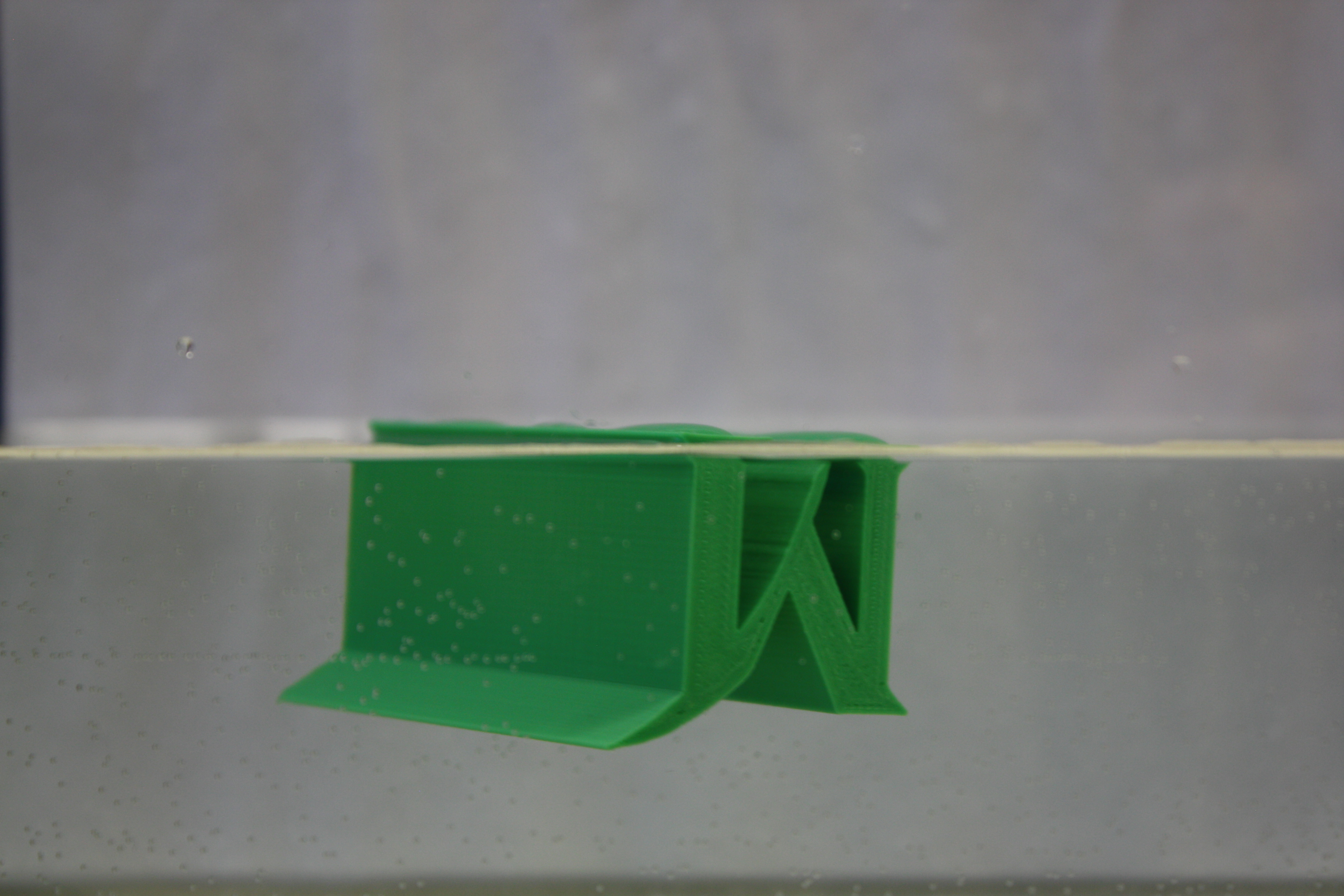}
\includegraphics[width=0.4\textwidth]{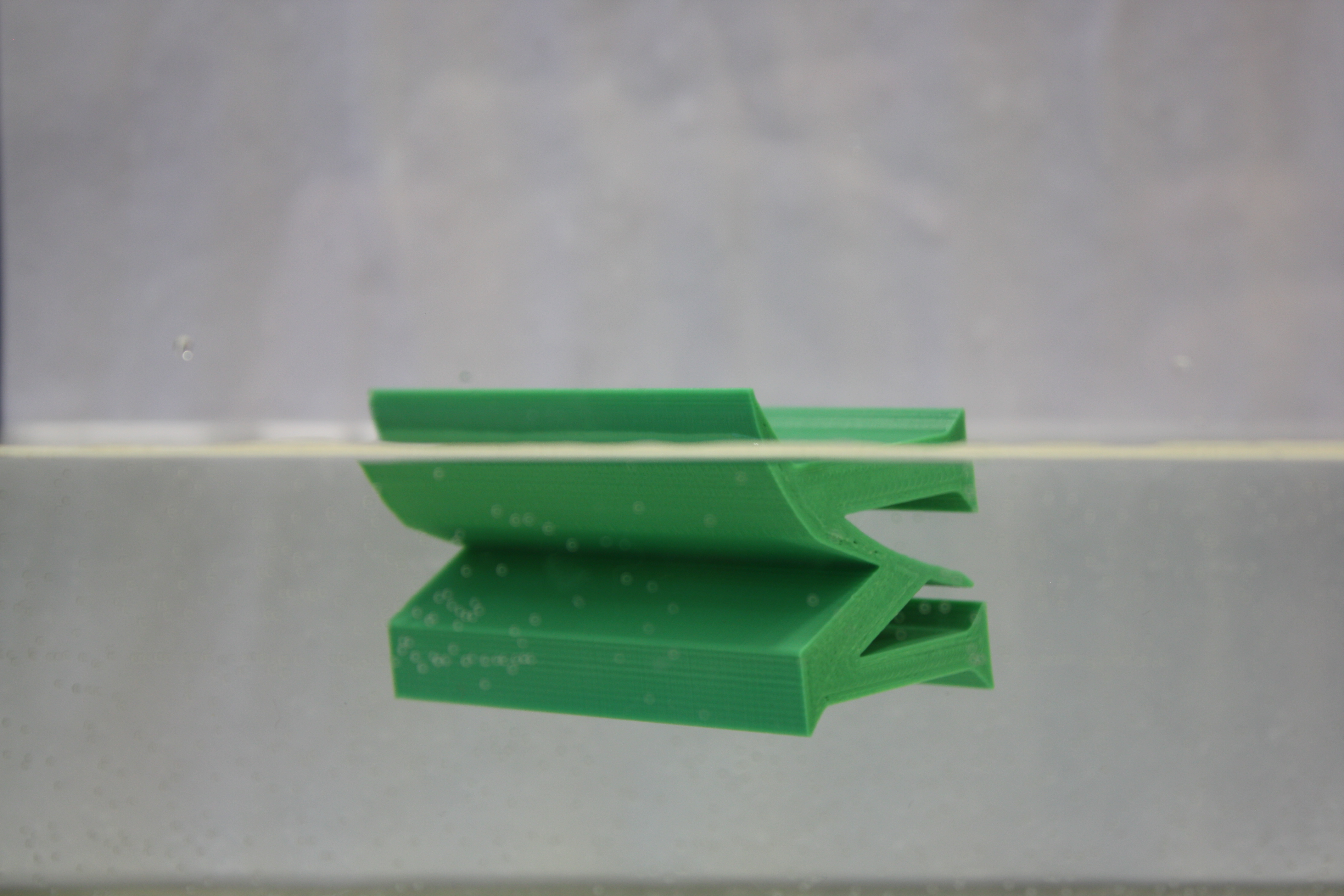}
\caption{The four stable orientations of the floating Mason M. The top left orientation matches up with the left most red dotted line in Figure~\ref{mason}. The top right orientation matches up with the second left-most red-dotted line in Figure~\ref{mason}. The bottom left orientation matches up with the second-right most red dotted line in Figure~\ref{mason}. The bottom right orientation matches up with the right-most red dotted line in Figure~\ref{mason}. }
\label{masonfloat}
\end{figure}

Figure~\ref{calving_MasonM} shows another variation of the full Mason M (with feathers) inspired by calving events that can happen with real icebergs.
These two plots show stable floating orientations, computed using the general polygonal code, where a {\em calving } event of shedding a feather results in a 
new stable floating orientation (right plot). The red line in the plot on the right shows the pre-calving waterline position. 

\begin{figure}[tb]
\includegraphics[width=0.4\textwidth]{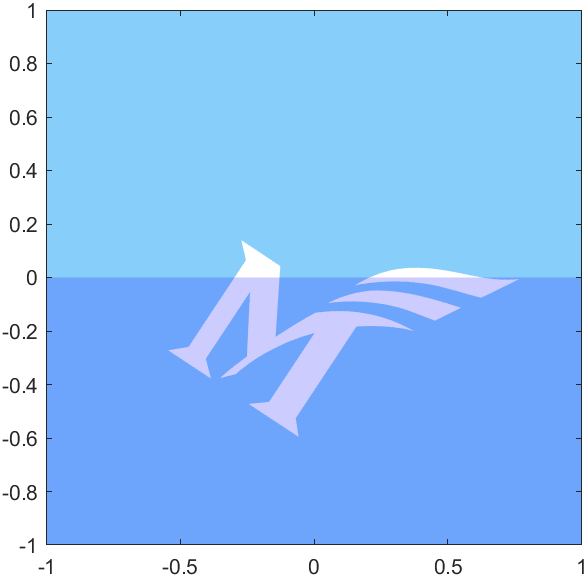}
\includegraphics[width=0.398\textwidth]{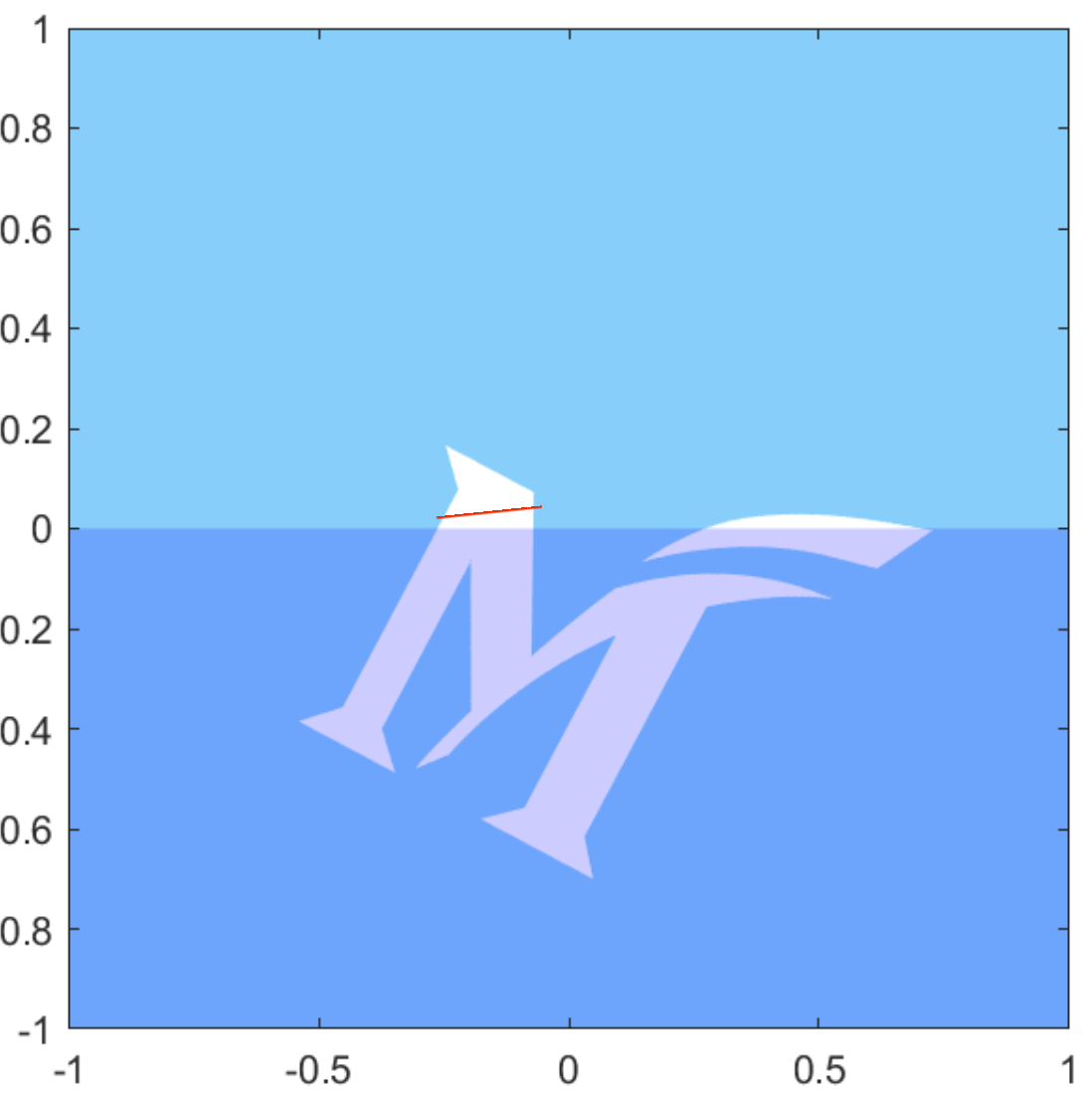}
\caption{The full Mason M (with feathers) and its floating orientation after a {\em calving} event where the upper feather is removed.  The red line shows pre-calving waterline.}
\label{calving_MasonM}
\end{figure}
%


\section{Conclusion and Open Problems}\label{sec:conclusion}

\subsection{Motivation}

 The equilibria of floating bodies have been studied since antiquity. Results on
 their stability associated with metacentric concepts that date back to
 eighteenth-century continue to play a major role in areas of naval architecture
 and the study of icebergs.

The field has recently gained mathematical interest, especially on Ulam's
floating body problem. In problem 19 from the Scottish Book Ulam asks: 
``Is a solid of uniform density which will float in water in every position a sphere?"
 \cite{ScottishBook}.
Counterexamples to Ulam's problem in the plane for $\rho = 1/2$ go
back to Auerbach \cite{Auerbach1938}, and for $\rho \neq 1/2$ to Wegner who also
obtained results for non-convex bodies (holes in the body are allowed) in
\cite{Wegner2020}. Most recently Florentin {\it et al.} \cite{Florentin2022}
gave an affirmative answer for a class of origin symmetric n-dimensional convex
bodies with $\rho = 1/2$ relative to water. While counterexamples in the
plane exist, Ulam's floating body problem in higher dimensions is open to the best of our 
knowledge.

\subsection{Limitations}

Using 3D printed models we were able to produce and collect experimental results
in agreement with the theory. Our results and experimentation process can be
reproduced readily, increasing access to 3D printers within educational
facilities gives way for students to experiment with stability and generate
different scenarios including Ulam's floating body problem.

 The biggest challenge faced when producing our 3D samples is related to
 density. To design infill patterns and infill densities using computer-aided
 design(CAD), we produced samples to measure precision and increase the accuracy
 of our experiment and identify possible sources of experimental errors. A
 limitation in the 3D printing process is the difficulty in producing
 low-density objects to desired accuracy.  To maintain an even distribution of
 mass, considerations must be taken with respect to the infill pattern. Zigzag,
 grid, triangle, and concentric patterns are recommended when an even
 distribution is desired. While Makerbot's cat infill model is not recommended
 for experimenting with homogenous objects as it is asymmetric, it made for
 interesting observations and ultimately motivated further investigation into
 non-homogeneous objects, and it's also really cute, cf. Figure~\ref{catinfill}.
 
 \begin{figure}[tb]
\includegraphics[width=0.3\textwidth]{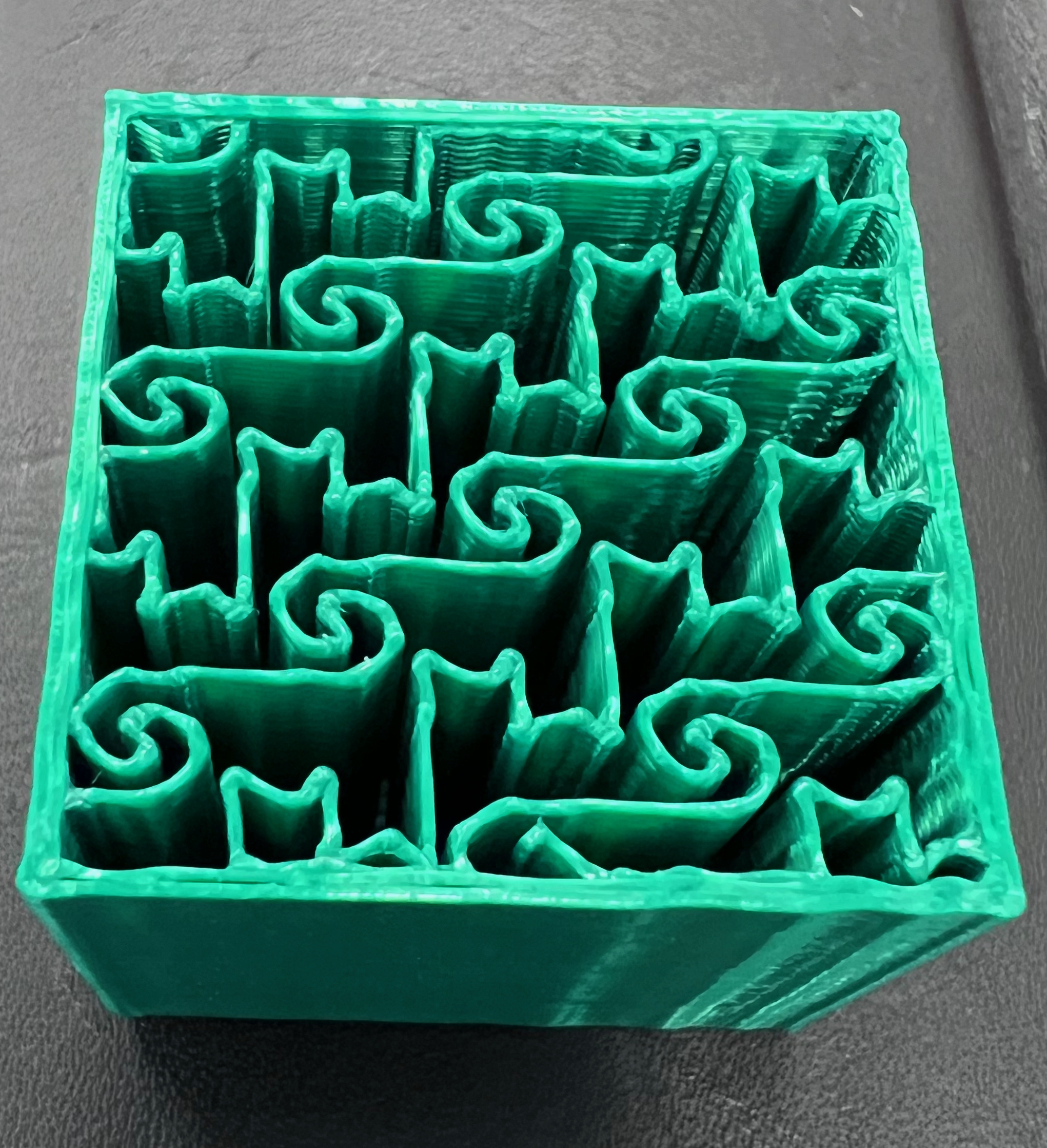}
\caption{The Makerbot cat infill.  }
\label{catinfill}
\end{figure}

 Our work here has focused on effectively two-dimensional shapes.   
 Mathematical challenges for 3D floating shapes have been examined 
 (e.g.~Erd\"{o}s {\it et al.} \cite{Erdos_etal1992_part2} and Wegner \cite{Wegner2009}) and of course most realistic floating objects such as
 icebergs are three dimensional.  Approaches using 3D Print design are likely to
 prove highly useful for future studies in these directions.

%
%
%
%

\section{Acknowledgements}

The authors would like to thank the referee for the helpful  comments that allowed us 
to improve this paper. 
We would also like to thank the Mason Experimental Geometry Lab (MEGL) and 
Mason's Math Maker Lab, particularly Maker Lab staff Patrick Bishop and MEGL member Will Howard, 
for supporting this project.  
We would also like to thank
the Department of Physics at George Mason University for the use of a digital scale. 
The research of E.S. was partially supported by the
Simons Foundation under Awards~636383.

%
%

\bibliographystyle{amsalpha}

\begin{thebibliography}{A}

%
%
%

\bibitem[Al]{Allaire1972}P.E. Allaire,
\textit{Stability of simply shaped icebergs},
Journal of Canadian Petroleum Technology
{\bf 11} (1972) 21--25.

\bibitem[Au]{Auerbach1938}H. Auerbach, 
\textit{Sur un probl{\`e}me de M. Ulam concernant l'{\'e}quilibre des corps flottants},
Studia Mathematica  {\bf 7} (1938) 121–-142. 

\bibitem[B]{Bailey1994}R.C. Bailey,
\textit{Implications of iceberg dynamics for iceberg stability estimation},
Cold Regions Science and Technology
{\bf 22} (1994) 197--203.

\bibitem[D]{Delbourgo1987}R. Delbourgo,
\textit{The floating plank},
Am. J. Phys. {\bf 55} (1987) 799--802.

\bibitem[DH]{DH2009}M.V. Deriabyn \& P.G. Hjorth,
\textit{Tip of the iceberg},
European Journal of Applied Mathematics
{\bf 20} (2009) 289--301.

\bibitem[E1]{Erdos_etal1992_part1}P. Erd\"{o}s, G. Schibler, \& R. C. Herndon,
\textit{Floating equilibrium of symmetrical objects and the breaking of symmetry.  Part 1: Prisms},
Am. J. Phys. {\bf 60} (1992) 335--345.

\bibitem[E2]{Erdos_etal1992_part2}P. Erd\"{o}s, G. Schibler, \& R. C. Herndon,
\textit{Floating equilibrium of symmetrical objects and the breaking of symmetry.  Part 2: The cube,
the octahedron, and the tetrahedron},
Am. J. Phys. {\bf 60} (1992) 345--356.

\bibitem[FF]{FF2021}Y. Feigel and N. Fuzailov,
\textit{Floating of a long square bar: experiment vs. theory},
Eur. J. Phys. {\bf 42} (2021) 035011.

\bibitem[F]{Florentin2022}D.I. Florentin, C. Sch\"{u}tt, E.M. Werner, \& N. Zhang,
\textit{Convex floating bodies of equilibrium},
Proceedings American Mathematical Society
https://doi.org/10.1090/proc/15697 (2022).

\bibitem[G]{Gilbert1991}E.N. Gilbert,
\textit{How things float},
Am. Math. Mon. {\bf 98} (1991) 201--216.

\bibitem[GIT]{GITHUB}D.M. Anderson, B.G. Barreto-Rosa, J.D. Calvano, L. Nsair,
and E. Sander, Supplementary Materials: Online code and files, {\tt github.com/danielmanderson/IcebergProject}. 

\bibitem[Ice1]{Iceberger}J. Tauberer, 
\textit{Iceberger},
{\tt https://joshdata.me/iceberger.html}. 

\bibitem[Ice2]{IcebergerRemix} chris@engaging-data.com, 
\textit{Iceberger Remixed},
{\tt https://engaging-data.com/iceberger-remixed/}.  

\bibitem[DAetal]{DAetal}
M. Lalegani Dezaki, M.K.A.M. Ariffin, A. Serjouei, A. Zolfagharian, S. Hatami, and M. Bodaghi,
\textit{Influence of Infill Patterns Generated by CAD and FDM 3D Printer on Surface Roughness 
and Tensile Strength Properties}, 
Applied Sciences 
{\bf 11:16} (2021) 7272.
doi: 10.3390/app11167272

\bibitem[M]{ScottishBook}
R. D. Mauldin,
\textit{The Scottish Book: Mathematics from The Scottish Caf\'{e} with Selected Problems from 
The New Scottish Book}, Second Edition,
Birkh\"{a}user, Switzerland, 2015.

\bibitem[MK]{MK2010}J. M\'{e}gel \& J. Kliava,
\textit{Metacenter and ship stability},
American Journal of Physics
{\bf 78} (2010) 738--747.
doi: 10.1119/1.3285975

\bibitem[OS]{OpenSCAD} 
\textit{OpenSCAD},
{\tt https://openscad.org}.  


\bibitem[Re]{Reid1963}W.P. Reid,
\textit{Floating of a long square bar},
Am. J. Phys. {\bf 31} (1963) 565--568.

\bibitem[Ro]{Rorres2004}C. Rorres,
\textit{Completing Book II of Archimedes's On Floating Bodies},
The Mathematical Intelligencer
{\bf 26} (2004) 32--42.

\bibitem[W1]{Wegner2009}F. Wegner,
\textit{Floating bodies of equilibrium in three dimensions.  The central symmetric case},
arXiv:0803.1043v2 [physics.class-ph] (2009).


\bibitem[W2]{Wegner2020}F. Wegner,
\textit{From elastica to floating bodies of equilibrium},
arXiv:1909.12596v4 [physics.class-ph] (2020).


\bibitem[WD]{WD1995}V. Wilczynski \& W.J. Diehl,
\textit{An alternative approach to determine a vessel's center of gravity: the center of buoyancy method},
Ocean Engng.
{\bf 22} (1995) 563--570.



\end{thebibliography}

\end{document}